%% file: tracerx.tex
\documentclass[acmsmall]{acmart}
\setcopyright{none}
\usepackage{booktabs}   %% For formal tables:
\usepackage{subcaption} %% For complex figures with subfigures/subcaptions
\usepackage{booktabs} % For formal tables
\usepackage{graphicx}
\usepackage{epstopdf}
\usepackage{tikz}
\usepackage{amsfonts}
\usepackage{amsmath}
\usepackage{listings}
\usepackage{url}
\usepackage{algpseudocode}
\usepackage{algorithm}
\usepackage{float}
\usepackage{color}
\usepackage{wrapfig}
\usepackage{stmaryrd}
\usepackage{color, colortbl}
\usepackage{pgf}
\usepackage{multirow}
\usepackage{balance}
\usepackage{longtable}
\usepackage{subcaption}
\input{commands}

\begin{document}

%% Title information
\title[TracerX:Dynamic Symbolic Execution with Interpolation]{TracerX: Dynamic Symbolic Execution with Interpolation}         %% [Short Title] is optional;
                                        %% when present, will be used in
                                        %% header instead of Full Title.
%\titlenote{with title note}             %% \titlenote is optional;
                                        %% can be repeated if necessary;
                                        %% contents suppressed with 'anonymous'
%\subtitle{Subtitle}                     %% \subtitle is optional
%\subtitlenote{with subtitle note}       %% \subtitlenote is optional;
                                        %% can be repeated if necessary;
                                        %% contents suppressed with 'anonymous'

%% Author information
%% Contents and number of authors suppressed with 'anonymous'.
%% Each author should be introduced by \author, followed by
%% \authornote (optional), \orcid (optional), \affiliation, and
%% \email.
%% An author may have multiple affiliations and/or emails; repeat the
%% appropriate command.
%% Many elements are not rendered, but should be provided for metadata
%% extraction tools.

%% Author with single affiliation.
\author{Joxan Jaffar}
%\authornote{with author1 note}          %% \authornote is optional;
                                        %% can be repeated if necessary
\orcid{0000-0001-9988-6144}             %% \orcid is optional
\affiliation{
  \institution{National University of Singapore}            %% \institution is required
}
\email{joxan@comp.nus.edu.sg}          %% \email is recommended

\author{Rasool Maghareh}
%\authornote{with author1 note}          %% \authornote is optional;
%% can be repeated if necessary
\orcid{0000-0002-8147-6590}             %% \orcid is optional
\affiliation{
	\institution{National University of Singapore}            %% \institution is required
}
\email{rasool@comp.nus.edu.sg}          %% \email is recommended

\author{Sangharatna Godboley}
%\authornote{with author1 note}          %% \authornote is optional;
%% can be repeated if necessary
\orcid{0000-0002-8147-6590}             %% \orcid is optional
\affiliation{
	\institution{National University of Singapore}            %% \institution is required
}
\email{sanghara@comp.nus.edu.sg}          %% \email is recommended

\author{Xuan-Linh Ha}
%\authornote{with author1 note}          %% \authornote is optional;
%% can be repeated if necessary
\orcid{0000-0003-1916-6812}             %% \orcid is optional
\affiliation{
	\institution{National University of Singapore}            %% \institution is required
}
\email{haxl@comp.nus.edu.sg}          %% \email is recommended

%\author{
%Joxan Jaffar \and
%Rasool Maghareh \and
%Sangharatna Godboley \and
%Xuan-Linh Ha}

%% Abstract
%% Note: \begin{abstract}...\end{abstract} environment must come
%% before \maketitle command
\begin{abstract}
\input{abstract}

\end{abstract}

%% 2012 ACM Computing Classification System (CSS) concepts
%% Generate at 'http://dl.acm.org/ccs/ccs.cfm'.

%% End of generated code

%% Keywords
%% comma separated list
  %% \keywords are mandatory in final camera-ready submission

%% \maketitle
%% Note: \maketitle command must come after title commands, author
%% commands, abstract environment, Computing Classification System
%% environment and commands, and keywords command.
\maketitle

%%%%%%%%%%%%%%%%%%%%%%%%%%%%%%%%%%%%%%%%%%%%%%%%%%%%%%%%%%%%%%%%%%%%%%%

\section{Introduction}
\label{sec:intro}
\input{intro}

%%%%%%%%%%%%%%%%%%%%%%%%%%%%%%%%%%%%%%%%%%%%%%%%%%%%%%%%%%%%%%%%%%%%%%

%\section{Motivating Examples (from FSE 2019)}
%\label{sec:example}
%\input{example}

\section{A Motivating Example}
\label{sec:motivating}
\input{motivating}

%%%%%%%%%%%%%%%%%%%%%%%%%%%%%%%%%%%%%%%%%%%%%%%%%%%%%%%%%%%%%%%%%%%%%%
%\vspace{-2mm}
\section{Background: Symbolic Execution}
\label{sec:symex}
\input{symex}

%\vspace*{-2mm}
\subsection{Symbolic Execution with Interpolation}
\label{sec:statesub}

\input{interpolation}

\section{The Main Algorithm}
\label{sec:mainalgo}

\input{algm}

%%%%%%%%%%%%%%%%%%%%%%%%%%%%%%%%%%%%%%%%%%%%%%%%%%%%%%%%%%%%%%%%%%%%%%
%\vspace{-2mm}
\section{A Practical Interpolation Algorithm}
\label{sec:algm1}

\input{algm1}

%%%%%%%%%%%%%%%%%%%%%%%%%%%%%%%%%%%%%%%%%%%%%%%%%%%%%%%%%%%%%%%%%%%%%%

%\vspace{-2mm}
\subsection{Remarks on Implementation}
\label{sec:impl-remark}
\input{impl_remark}

%%%%%%%%%%%%%%%%%%%%%%%%%%%%%%%%%%%%%%%%%%%%%%%%%%%%%%%%%%%%%%%%%%%%%%

\section{Experimental Evaluation}
\label{sec:imp_and_eval}
\input{main_experiment}

\subsection{Supplementary Experiment (Exploration)}
\label{sec:supl-expr}
\input{supl_experiment}

%%%%%%%%%%%%%%%%%%%%%%%%%%%%%%%%%%%%%%%%%%%%%%%%%%%%%%%%%%%%%%%%%%%%%%

\section{Summary of the Results}
\label{sec:summ-results}

\input{summ_experiment}

\section{Related Work}
\label{sec:related}
\input{related}

%%%%%%%%%%%%%%%%%%%%%%%%%%%%%%%%%%%%%%%%%%%%%%%%%%%%%%%%%%%%%%%%%%%%%%

\section{Conclusion}
\label{sec:conclusion}
\input{conclusion}

%%%%%%%%%%%%%%%%%%%%%%%%%%%%%%%%%%%%%%%%%%%%%%%%%%%%%%%%%%%%%%%%%%%%%%

%% Acknowledgments
%\begin{acks}                            %% acks environment is optional
%                                        %% contents suppressed with 'anonymous'
%  %% Commands \grantsponsor{<sponsorID>}{<name>}{<url>} and
%  %% \grantnum[<url>]{<sponsorID>}{<number>} should be used to
%  %% acknowledge financial support and will be used by metadata
%  %% extraction tools.
%  This material is based upon work supported by the
%  \grantsponsor{GS100000001}{National Science
%    Foundation}{http://dx.doi.org/10.13039/100000001} under Grant
%  No.~\grantnum{GS100000001}{nnnnnnn} and Grant
%  No.~\grantnum{GS100000001}{mmmmmmm}.  Any opinions, findings, and
%  conclusions or recommendations expressed in this material are those
%  of the author and do not necessarily reflect the views of the
%  National Science Foundation.
%\end{acks}

\clearpage
%% Bibliography
\bibliography{references}

%% Appendix
\clearpage
\appendix
%%%%%%%%%%%%%%%%%%%%%%%%%%%%%%%%%%%%%%%%%%%%%%%%%%%%%%%%%%%%%%%%%%%%%%

\section{Appendix}
\label{sec:appendix}
\input{appendix}

%%%%%%%%%%%%%%%%%%%%%%%%%%%%%%%%%%%%%%%%%%%%%%%%%%%%%%%%%%%%%%%%%%%%%%

\end{document}

%% file: commands.tex
\newcommand{\tracerx}{\mbox{\sc tracer-x}}

\newcommand{\ignore}[1]{}

\newcommand{\stuff}[1]{
        \begin{minipage}{6in}
        {\tt \samepage
        \begin{tabbing}
        \hspace{3mm} \= \hspace{3mm} \= \hspace{3mm} \= \hspace{3mm} \= \hspace{3mm} \= \hspace{3mm} \=
\hspace{3mm} \= \hspace{3mm} \= \hspace{3mm} \= \hspace{3mm} \= \hspace{3mm} \= \hspace{3mm} \= \hspace{3mm} \= \kill
        #1
        \end{tabbing}
       }
        \end{minipage}
}

\newcommand{\trans}{\longrightarrow}

\newcommand{\loc}{\mbox{$\ell$}}
\newcommand{\locations}{\mbox{$\Sigma$}}

\newcommand{\symstate}{\mbox{$s$}}

\newcommand{\pc}{\mbox{\loc}}

\newcommand{\pcstart}{\mbox{$\loc_{\textsf{start}}$}}

\newcommand{\pathcond}{\mbox{$\Pi$}}
\newcommand{\store}{\mbox{$\sigma$}}

\newcommand{\mapstatetoformula}[1]{\mbox{$\llbracket {#1} \rrbracket$}}

\newcommand{\newtransition}[3]{#1 \xrightarrow{#3} #2}

\newcommand{\typevar}{\mbox{\emph{Vars}}}

\newcommand{\typestmt}{\mbox{\emph{Stmts}}}

\newcommand{\true}{\mbox{\frenchspacing \it true}}
\newcommand{\false}{\mbox{\frenchspacing \it false}}

\newcommand{\hstore}{\mbox{$\mu$}}
\newcommand{\eval}[3]{\llbracket {#1} \rrbracket_{#2}^{#3}}
\newcommand{\evalTwo}[2]{\llbracket {#1} \rrbracket_{#2}}

\newcommand{\Intpsymbol}{\mbox{$\Psi$}}
\newcommand{\tracer}{\mbox{\sc tracer}}

\newcommand{\triple}[3]{\langle{#1},{#2},{#3}\rangle}

\newcommand{\pp}[1]{{\color{blue}\mbox{$\langle{#1}\rangle$}}}
\newcommand{\prpt}[1]{\mbox{{\color{blue}{$\langle$#1$\rangle$}}}}

\newcommand{\benchmark}[1]{\texttt{\small #1}}

\newcommand{\stmt}{\mbox{\it {\frenchspacing stmt}}}
\newcommand{\abduction}{\mbox{\sc {\frenchspacing abduction}}}
\newcommand{\interpolant}{\mbox{\sc intp}}
\newcommand{\backprop}{\mbox{\sc BackProp}}
\newcommand{\core}{\mbox{\sc core}}
\newcommand{\separate}{\mbox{\sc separate}}

%% file: abstract.tex
Dynamic Symbolic Execution is an important method for the
testing of programs.  An important system on DSE is
KLEE~\cite{cadar08klee} which inputs a C/C++ program annotated with
symbolic variables, compiles it into LLVM, and then emulates the
execution paths of LLVM using a specified backtracking strategy.  The
major challenge in symbolic execution is {\em path explosion}.  The
method of \emph{abstraction
learning} \cite{jaffar09intp,mcmillan10lazy,mcmillan14lazy} has been
used to address this.  The key step here is the computation of
an \emph{interpolant} to represent the learnt abstraction.

~~~~~~~In this paper, we present a new interpolation algorithm and implement it
on top of the KLEE system.   The main objective is to address the path explosion
problem in pursuit of \emph{code penetration}: to prove that a target program
point is either reachable or unreachable.
That is, our focus is verification.
We show that despite the overhead of computing interpolants,
the \emph{pruning} of the symbolic execution tree
that interpolants provide often brings significant overall
benefits.   
We then performed a comprehensive experimental evaluation against KLEE,
as well as against one well-known system that is based on Static Symbolic
Execution,  CBMC \cite{cbmc}.
Our primary experiment shows code penetration success at a new level,
particularly so when the target is hard to determine.
A secondary experiment shows that our implementation is competitive
for testing.

%% file: intro.tex
Symbolic execution (SE) has emerged as an important method to reason
about programs, in both verification and testing.  By reasoning about
inputs as symbolic entities, its fundamental advantage over
traditional black-box testing, which uses concrete inputs, is simply
that it has better \emph{coverage} of \emph{program paths}.  In
particular, \emph{dynamic symbolic execution} (DSE), where the
execution space is explored \emph{path-by-path}, has been shown
effective in systems such as DART~\cite{godefroid05dart},
CUTE~\cite{sen05cute} and KLEE~\cite{cadar08klee}\footnote{It is not
universally agreed that KLEE is a DSE system, but we follow the
terminology in two recent CACM articles~\cite{avgerinos16veritesting,cadar13se}.}.

A key advantage of DSE is that by examining a single path, the
analysis can be both precise (for example, capturing intricate details
such as the state of the cache micro-architecture), and efficient (for
example, the constraint solver often needs to deal with path
constraints that are aggregated into a single conjunction).  
Another advantage is the
possibility of reasoning about system or library functions which we
can execute but not analyze, as in the method of \emph{concolic
testing} CUTE~\cite{sen05cute}.  Yet another advantage is the ability to realize
a \emph{search strategy} in the path exploration, such as in a
random, depth-first, or breadth-first manner, or in a
manner determined by the program structure.  However, the key
disadvantage of DSE is that the number of program paths is in general
exponential in the program size, and most available implementations of
DSE do not employ a general technique to prune away some paths.
Indeed, a recent paper \cite{avgerinos16veritesting} describes that
DSE ``traditionally forks off two executors at the same line, which
remain subsequently forever independent'', clearly suggesting that the
DSE processing of different paths have no symbiosis.

A variant of symbolic execution is that of Static Symbolic Execution (SSE), 
see e.g. \cite{khurshid2003generalized,avgerinos16veritesting}.
The general idea is that the symbolic execution tree is
encoded as a single logic formula whose treatment can be outsourced to an SMT
solver \cite{de2002lemmas}.
\ignore{
	Another related area which employs pruning for expedite bug finding is  bounded 
	model checking (BMC); see e.g. the systems CBMC \cite{clarke2004tool} and 
	LLBMC \cite{falke2013llbmc}. This methodology is also known as \emph{static} symbolic 
	execution (SSE), where the (whole) program is encoded by an SMT solver \cite{de2002lemmas}.
}
% The exploration process is delegated away to the 
The solver then deals with 
% has to deal with the general problem of reasoning about 
what is essentially a huge disjunctive formula. 
Clearly, there are  some limitations to this approach as compared with DSE,
for example, the loop bounds, including nested loops, must be pre-specified.
However, SSE has a huge advantage over (non-pruning) DSE: its SMT solver
can use the optimization method of \emph{conflict directed clause learning} (CDCL) 
\cite{marques1999grasp}.
See e.g. Section 3.4 of \cite{de2008proofs} on how the SMT solver Z3 exploits CDCL.
Essentially, CDCL enables ``pruning'' in the exploration process of the solver. 

In this paper, our primary objective is to address the path explosion
problem in DSE.  More specifically, we wish to perform a path-by-path
exploration of DSE to enjoy its benefits, but we include a
\emph{pruning mechanism} so that path generation can be eliminated if
the path generated so far is guaranteed not to violate the
stated safety conditions.
Toward this goal, we employ the method of \emph{abstraction
learning} \cite{jaffar09intp}, which is more popularly known
as \emph{lazy annotations\/} (LA) \cite{mcmillan10lazy,mcmillan14lazy}.
The core feature of this method is the use of \emph{interpolation},
which serves to generalize the context of a node in the symbolic
execution tree with an approximation of the weakest precondition of
the node.  This method has been implemented in the \tracer{} system
~\cite{jaffar11unbounded,jaffar12tracer} which was the first
system to demonstrate DSE with pruning.  While \tracer{} was able to
perform bounded verification and testing on many examples, it could
not accommodate industrial programs that often dynamically manipulate
the heap memory.  Instead, \tracer{} was primarily used
to evaluate new algorithms in verification, analysis, and testing,
e.g., ~\cite{chu16cache,chu12symmetry,jaffar13test}.

The main contribution of this paper is the design and implementation
of a new interpolation algorithm, and integration into the KLEE
system.  In our primary experiment, we compare against KLEE.  In a
secondary experiment, we consider the related area of Static Symbolic
Execution (SSE).  The reason is that, while SSE is generally
considered as significantly different from DSE, SSE is a competitor to
DSE because they both address many common analysis problems.  

Our main experimental result is that our algorithm leads in \emph{code penetration}:
given a target, which is essentially a designated program point, 
can one prove that the target is reachable, or prove that the target is unreachable.
Thus our algorithm is more aligned with verification rather than testing.
We suggest two driving applications for such verification.
One is to confirm/deny an ``alarm'' from a static analysis; an alarm is a target
for which there is a plausible reason for it to be a true bug.
Another application is dead code detection.  Penetration addresses both these questions.

Given that our algorithm is relatively heavy-weight, it is expected that
for some examples, the overhead is not worth it.
We show firstly that our implemented system performs well on a large
benchmark suite.   We then considered a subset of the original targets called \emph{hard targets}.
These are obtained by filtering out targets
that can be proved easily by state-of-the-art methods: vanilla symbolic execution
for reachable targets, and static analysis for unreachable targets.
We then show that for the remaining (hard) targets, the performance gap widens.

We then performed a secondary experimental evaluation for testing.
Here we followed the setup of the TEST-COMP competition by evaluating
on \emph{bug finding} (given a set of targets, find one) and \emph{code
coverage} (where all program blocks are targets, and to find as many
as possible).  In this secondary experiment, we show that our implementation
is competitive.

%% file: motivating.tex
Consider the shortest path problem in graphs.
See Fig. \ref{fig:motive} where we assume the graph has edges, where each of which has a ``distance''.
The edges point from a lower-numbered vertex to a higher one.
%Each path in the symbolic execution tree
%represents a traversal of a path in a graph from source node $1$ to target node $N$.
The variable $d$ computes the distance between nodes $1$ and $N$ 
via the traversed path. In the end, we want to know if $d \geq BOUND$ for some given constant $BOUND$.
Symbolic execution on this program (where the variable {\tt node} is symbolic)
will traverse all paths, and therefore will be able to check the bound.
Then, by iteratively executing symbolic execution with various values of $BOUND$,
we can solve the shortest path problem.

\begin{figure}[!htb]
	\small{
		\stuff{
			\#define N ... \\
			int graph[N+1][N+1] = \{ ... \}; \\
			int node = 1, d = 0; \\
			while (node < N) \{ \\
			\>    int next = new symbolic variable; \\
			\>     for (i = node + 1; i < N; i++) \\
			\>\>        if (next == i) { d += graph[node][next]; break; } \\
			\>   node = next; \\
			\} \\
			assert (d >= BOUND);
		}
	}
	\caption{Motivating Example 1}
	\label{fig:motive}
	\vspace{-4mm}
\end{figure}

\begin{minipage}{\textwidth}
	\begin{minipage}[b]{0.55\textwidth}
		\begin{center}
			\includegraphics[scale=0.55]{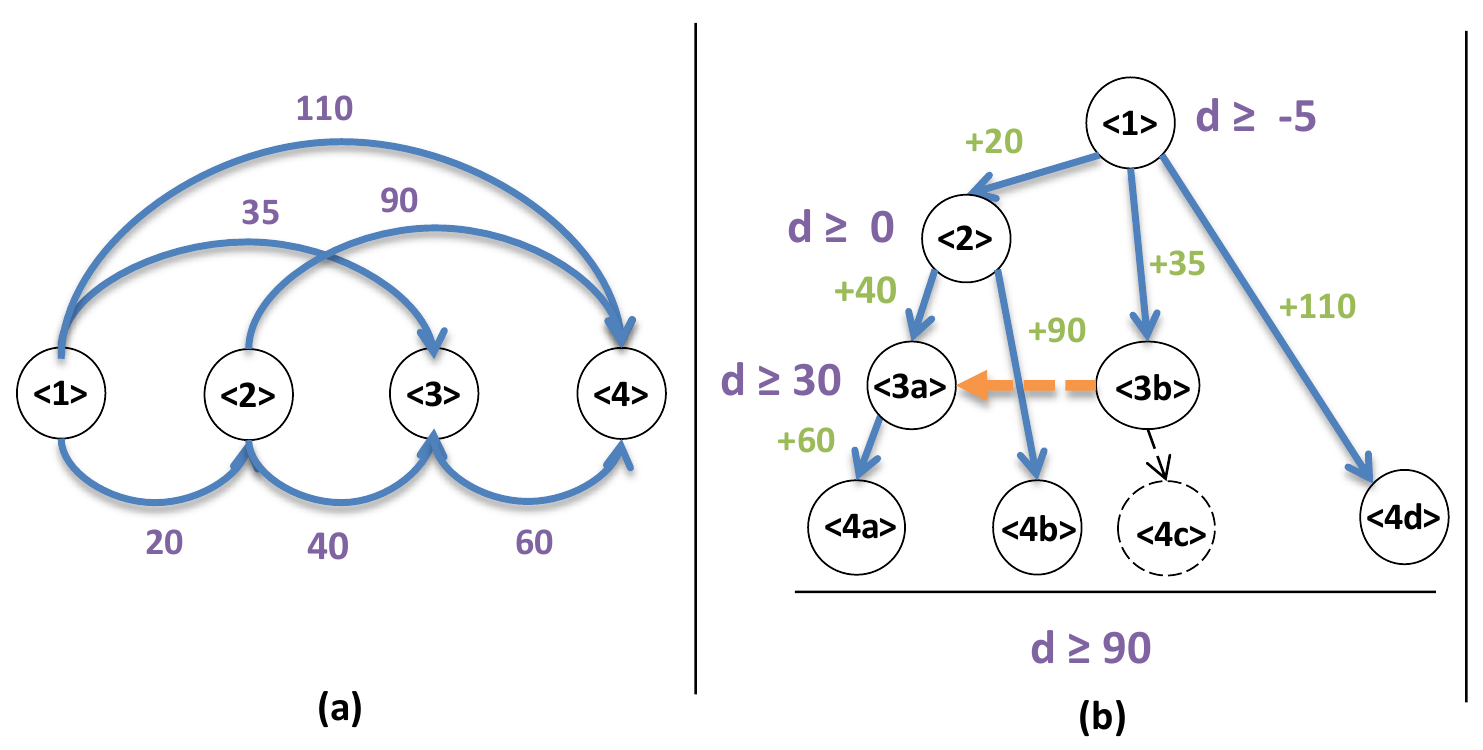}
			\vspace*{-6mm}
		\end{center}
		\captionof{figure}{Example with Interpolation (a) Graph (b) Symbolic Execution Tree }
		\label{fig:motive-graph}
	\end{minipage}
	\hfill
	\begin{minipage}[b]{0.44\textwidth}
		\centering
		\vspace*{3mm}
		\begin{tabular}{c|c|c|c|c|}
			\cline{2-5} 
			4 & $\infty$ & $\infty$ & $\infty$  & $\infty$  \\ \cline{2-5} 
			3 & $\infty$  & $\infty$ & $\infty$  & 60 \\ \cline{2-5}  
			2 & $\infty$  & $\infty$  & 40 & 90 \\ \cline{2-5} 
			1 & $\infty$  & 20 & 35 & 110 \\ \cline{2-5} 
			\multicolumn{1}{c}{} & \multicolumn{1}{c}{1}  & \multicolumn{1}{c}{2} & \multicolumn{1}{c}{3} & \multicolumn{1}{c}{4} \\ 
		\end{tabular}
		\vspace{10mm}
		\captionof{table}{Adjacency Matrix}
		\label{tab:motive-matrix}
	\end{minipage}
	
\end{minipage}

\noindent\begin{example}\label{ex:motivEx}
Consider the example graph in Fig. \ref{fig:motive-graph} (a) and its 
respective adjacency matrix in Table  \ref{tab:motive-matrix}. The 
execution tree of the program from Fig. \ref{fig:motive} is presented 
in Fig. \ref{fig:motive-graph}(b).  The program points \prpt{1}{} and \prpt{4}{} 
represent the source and destination nodes.  
The program points using the same number, e.g. \prpt{3a} and \prpt{3b}, 
identify different visits to the same program point.
Now assume we want to ensure that the distance between nodes $1$
and $4$ is greater than $BOUND = 90$ at the end. 
A basic DSE tool will traverse all the paths in the execution tree 
if there is no pruning.  There would be a problem for this example
because the number of the paths is exponential in $N$.

\ignore{
On the other hand,  SSE-based approaches possess
a pruning technology. These tools require unrolling the loops 
statically. Since we generally can be at different nodes of the
graph in different iterations, the loop iterations are in general
not comparable.  Hence, SSE-based approaches are not able to 
prune the search space effectively either.
}

Now we will demonstrate how our approach can prune the execution tree.
Since our approach is built on top of DSE, all the program points in the loop
iterations possess the same program point. Hence, we can compare the states
in different iterations and prune them when possible.  Consider the 
leftmost path \prpt{1}, \prpt{2}, \prpt{3a}, and \prpt{4a}. 
%\textcolor{red}
%{Sanghu's-Sugg:I think this path is the "right most" for the execution
%tree. Also, this is a true branch path from \prpt{1}. I am just
%wondering whether you want to select the path \prpt{1} to \prpt{4d}
%first (this is the most left path). Please have a look, if everything
%is ok, please ignore it.} 
At the end of this path, since the distance
$d$ is 120, the assertion is not violated and an \emph{interpolant} is
stored at \prpt{3}: $d \geq 30$ 
%\textcolor{red}{Sanghu's-Sugg:"stored
%at \prpt{3}: d >= 30" ---> check whether \prpt{3} should
%be \prpt{3a}?}. 
This interpolant represents the \emph{weakest
precondition} at \prpt{3} which satisfies the assertion. It is
computed by updating the safety property $d \geq 90$ with the update
on $d$ between nodes \prpt{3a}, and \prpt{4a}: $d = d + 60$.  Next,
this interpolant is passed to the parent node \prpt{2} considering the
update on the variable $d$ between \prpt{2} and \prpt{3a}: $d \geq
-10$.

Backtracking back to \prpt{2}, we now consider the second path \prpt{1}, \prpt{2}, \prpt{4b}.
Here again, the assertion is not violated and an interpolant is passed to node  \prpt{2}
considering the update on $d$ between  \prpt{2} and \prpt{4b}:  $d \geq 0$.
Now the intersection of the two interpolants received from the successor nodes of \prpt{2} is 
stored as the interpolant of \prpt{2}: $d \geq 0$.

Moving on,  this interpolant is updated and sent to the parent node \prpt{1}:  $d \geq -20$.
Next, the path \prpt{1}, \prpt{3b}, and \prpt{4c} is traversed. Here, at node \prpt{3b} the
distance is $d = 35$, and the interpolant stored at \prpt{3a} ($d \geq 30$) can be
used to \emph{prune} this node. Note importantly that node  \prpt{3a} was visited in the third 
iteration of the loop (in the program from Fig. \ref{fig:motive}) and node \prpt{3b} is 
visited in the second iteration. Here, the pruning is sound since both program points
are on the same node in the graph (node $3$) and the state at \prpt{3b} satisfies the 
respective weakest precondition interpolant.  
% This pruning is not possible in SSE-based approaches.

This subsumed node is considered to have an interpolant, $d \geq 30$, that which subsumed it.
This is passed to node \prpt{1}: $d \geq -5$.
Finally, the last path \prpt{1}, \prpt{4d} is traversed and since the assertion is 
not violated an interpolant is passed to node \prpt{1}: $d \geq -20$. 
Now the three interpolants passed to \prpt{1},
$d \geq -20$ from \prpt{2}, 
$d \geq -5$ from \prpt{3b}, 
$d \geq -20$ from \prpt{4c}, 
is combined as a conjunction and stored  as an interpolant at \prpt{1}: $d \geq -5$. 

This interpolant at the root infers that all values of $d$ greater than or equal
$-5$ will not violate the assertion. From this, we can infer that if the condition in 
the assertion had been $d \geq 95$, it would still have not been violated. 
That is, we can,in fact, conclude that
the shortest path between nodes $1$ and $4$ in the graph is $95$, a tighter bound than what we started with.
$\Box$
\end{example}

For this example, our implementation can deal with 1000 nodes, while CBMC and KLEE tops out at 100 and 25 respectively.

%\begin{figure}[!htb]
%\begin{center}
%\vspace*{-3mm}
%\includegraphics[scale=0.65]{image/motiv}
%\end{center}
%\caption{Example with Interpolation (a) Graph (b) Symbolic Execution Tree (c) Adjacency Matrix}
%\label{fig:motive}
%\end{figure}

%\begin{figure}[b]
%(a) PICTURE----(b) \begin{tabular}{|c|c|c|c|}
%\hline
%$\infty$ & $\infty$ & $\infty$  & $\infty$  \\ \hline 
%$\infty$  & $\infty$ & $\infty$  & 90 \\ \hline 
%$\infty$  & $\infty$  & 40 & 90 \\ \hline 
%$\infty$  & 20 & 25 & 30 \\ \hline 
%\end{tabular}(c) SYMEX TREE
%\caption{Example Graph}
%\label{fig:motive}
%\end{figure}

%Assert(distance >= 90)
%
%Path 20->40 (intp is d >= 30) subsumes 30
%After interpolation, at root, distance >= -5
%So can conclude shortest path is 85.

%% file: symex.tex
We formalize dynamic symbolic execution (DSE) for a toy programming language.
The variables in a program are denoted \typevar.
Other than the program variables $\typevar_p$, there
are also \emph{symbolic} variables $\typevar_s$.
A basic statement is the assignment, $var$ = $exp$
where $exp$ is some arithmetic or boolean expression.
Another basic statement is {\tt assume($exp$)} where
$exp$ is a boolean expression.
Note that ``assertions'' can be modeled by {\tt assume()} statements
coupled with a distinguished statement representing {\tt error}.

For brevity, we omit other statements,
e.g. functions, memory operations like arrays, malloc, etc.
Extension to cover theses instructions would be routine. 
In general, we follow the semantics used by KLEE for 
these instructions. We will further discuss some implementation details 
for these instructions in Section \ref{sec:impl-remark}.

%a subset language of
%LLVM.  Our formalization is very much similar to other KLEE-like
%systems~\cite{avgerinos16veritesting,cadar08klee,kuznetsov12efficient}.
%the key difference is our treatment of {\tt malloc}, which enables an
%algorithm for ``lazy annotation'' described in later Sections.
%
%Though
%being simple, the language is enough to demonstrate the difficulties
%to apply ``lazy annotation'' when dealing with heap-based
%programs.
%{\color{red} Malloc is different}
%
We model a program P by a transition system: a tuple
$\triple{\locations}{\pcstart}{\trans}$ where $\locations$ is
the set of program points and $\pcstart \in \locations$ is the
\emph{unique} initial program point.
%
% and $O \subseteq \locations$ denotes the set of terminal program points.
%
Let $\trans \subseteq \locations \times \locations \times \typestmt$,
where $\typestmt$ is the set of program statements, be the transition
relation that relates a state to its (possible) successors by
executing the statements. We shall use
$\newtransition{\loc}{\loc'}{\texttt{stmt}}$ to denote a transition
relation from $\loc \in \locations$ to $\loc' \in \locations$
executing the statement $\texttt{stmt} \in \typestmt$.

%These are used in formulas describing states, but not in the programs themselves
%(a.k.a. ``ghost'' variables).
% Let $\tilde{I}$ denote the symbolic inputs, which are \emph{logical} variables.  
%For brevity, we omit the presentation of the deallocation
%operation ``free'' because its consideration is straightforward yet
%very dependent on specific implementation details. 
%(In fact, \tracerx{} implements it.)
%\renewcommand{\arraystretch}{1.2}
%\vspace*{-3mm}
%\begin{table}[!ht]
%\begin{center}
%	
%\begin{tabular}{lll}
%%  $|$ store($exp$, $exp$) $|$	
%% $var$ := $malloc(c)$ $|$ 
%$stmt$ & ::= & $var$ := \mbox{input()} $|$ $var$ := $exp$  $|$ assume($exp$) $|$ error $|$ halt  \\
%
%%load($exp$) $|$
%$exp$ & ::= &   $exp$ $\Diamond_{b}$ $exp$ $|$ $\Diamond_{u}$ $exp$ $|$ $var$ $|$ $c$ \\
%
%
%$\Diamond_{b}$ & ::= & typical binary operators \\
%$\Diamond_{u}$ & ::= & typical unary operators \\
%
%$c$   & ::= & 32-bit unsigned integer \\
%$var$  & $\in$  & $\typevar$ \\
%\end{tabular}
%\caption{A Simple Intermediate Language}
%\label{table:lang}
%\end{center}
%\vspace*{-8mm}
%\end{table}
%\noindent
%For presentation purposes, we consider only expressions 
%that evaluate to 32-bit integer values, and memory
%objects are aligned in 4-bytes.   
%Generalizing to additional types is straightforward. 

%We also assume that the input programs are well-typed in
%the obvious way.

%We now define the notion of symbolic state. 
%In addition to a (stack) 
%store in traditional definition (e.g. in \cite{jaffar12tracer}), 
%we also track a (symbolic) ``heap store''.

\vspace*{-1mm}
\begin{definition}[Symbolic State] 
\label{def:symstate}
A symbolic state $s$ is a tuple $\langle \pc, \pathcond \rangle$, 
where $\pc \in \locations$ is the current program
point, 
%the \emph{store} $\store$ is a map from program variables to
%terms (expressions which define ``values''), 
%the \emph{heap store} $\hstore$ is a map from addresses to terms, 
and the constraint store (or ``context'') $\pathcond$ is a first-order formula
over symbolic variables $\typevar_s$ and program variables $\typevar_p$. 
%symbolic inputs that the inputs must satisfy in order for the
%execution to reach the current state. 
\hfill $\Box$
\end{definition}
\vspace*{-2mm}

%Let $\mempty$ denote an \emph{empty} map. Given a map $M$, let $dom(M)$ be
%the domain of $M$. From Definition~\ref{def:symstate}, for every symbolic state
%$\langle \pc, \store, 
%%\hstore,
% \pathcond \rangle$, $Vars \equiv
%dom(\store)$. Let \mbox{$M[key \mapsto val]$} denote the most basic operation
%of a map: if $key$ does not exist in $M$, it will be added; and then
%the associated value for $key$ in $M$ is set to $val$.
%We also define the binary relation $\msubsume$ between
%maps as follows: $M_1 \msubsume M_2$ if $dom(M_1) \subseteq dom(M_2)$
%and $\forall x \in dom(M_1) \cdot M_1(x) \equiv M_2(x)$.

%{\hstore}
The \emph{evaluation} $\eval{x = exp}{\pathcond}{}$ of
an expression $exp$ with the constraint store $\pathcond$ is 
defined in the standard way by $\pathcond[x/exp]$.  
Similarly, the evaluation $\eval{assume(exp)}{\pathcond}{}$ is 
defined by $\pathcond \wedge exp$.   
(Note that our expressions have no \emph{side-effects}.)
% Calls to system functions must be treated differently,
%such as in the case of \texttt{malloc}.
%  For example,
%$\eval{v}{\store}{} = \store(v)$ (if $v$ is a program
%variable), $\eval{c}{\store}{} = c$ (if $c$ is an integer), and 
%$\eval{e_1~\Diamond_{b}~e_2'}{\store}{}
%= \eval{e_1}{\store}{}~\Diamond_{b}~{\eval{e_2}{\store}{}}$
%(where $e_1, e_2$ are expressions and $\Diamond_{b}$ is a binary
%operator). 
%$\eval{load(e)}{\store}{} = \hstore(\eval{e}{\store}{})$.  
The notion of evaluation is
extended for a set of constraints in an intuitive way. 
The evaluation of the constraint store of a state is denoted by $\eval{s}{}{}$. 

A symbolic state $\symstate \equiv \langle \pc, \pathcond \rangle$ is
called \emph{infeasible}
if $\eval{s}{}{}$ is unsatisfiable. Otherwise, the state is
called \emph{feasible}; symbolic execution is possible
from a feasible state only.  
%
%Let the symbolic states be denoted by \typesymbstate.

% 
%The set of first-order formulas and 
%symbolic states are denoted by \typefo\
%and \typesymbstate, respectively.

\vspace*{-1mm}
\begin{definition}[Transition Step]
\label{def:trans}
Given a feasible symbolic
 state
 $s \equiv \langle \pc, \pathcond \rangle$, % \in \typesymbstate$,
and 
a transition system
 $\triple{\locations}{\pcstart}{\trans}$ 
 the symbolic execution of transition $\pc \xrightarrow{\sf
 stmt} \pc'$ returns a \emph{successor state}
 $\langle \pc',  \pathcond' \rangle$ where
 $\pathcond'$ is computed via $\eval{stmt}{\pathcond}{}$.
% in  Table~\ref{table:symex}. 
$\Box$
\end{definition}
\vspace*{-1mm}

%\begin{table}[!tbh]
%\vspace*{-3mm}
%%\hspace*{-1mm}
%\small
%\centering
%\begin{tabular}{|lll|}
%\hline
%\texttt{stmt} & $\store'$ &  $\pathcond'$ \\
%\hline
%\hline
%$v$ := \mbox{input()} & $\store$[$v \mapsto i$] & $\pathcond$ \\
%% \multicolumn{4}{|r|}{where $i$ is fresh} \\
%% & { symbolic variable} \\
%\hline
%$v$ := $e$ & $\store$[$v \mapsto \eval{e}{\store}{}$]~ & $\pathcond$ \\
%\hline
%%store($e_1, e_2$) & $\store$ & $\hstore$[$\eval{e_1}{\store}{\hstore} \mapsto \eval{e_2}{\store}{\hstore}$] & $\pathcond$ \\
%%\hline
%assume(e) & $\store$ & $\pathcond \wedge \eval{e}{\store}{\hstore}$ \\
%\hline
%%$v$ := malloc($c$) &  $\store$[$v \mapsto a$] & 
%%$\hstore$[$a \!\mapsto \!\_$]$\cdots$[$a\!+\!c/4\!-\!1\!\mapsto \!\_$] ~&
%% $\pathcond$ \\
%%$\pathcond \wedge \dom(\hstore') = [a,\!a\!+\!c/4\!-\!1] \uplus dom(\hstore)$ \\
%% \multicolumn{4}{|c|} 
%%{where $a$ is the address returned from concretely executing malloc($c$)} \\
%% & &  {where $a$ is a \emph{fresh} symbolic variable.}
%% \\
%%\hline
%\end{tabular}
%\vspace{2mm}
%\caption{Operational Semantics for Symbolic Execution, Note that variable $i$  is a fresh symbolic variables.}
%\label{table:symex}
%\vspace*{-6mm}
%\end{table}

\noindent
% es; we shall henceforth call them \emph{ghost variables}.  
%We also add into the path condition the following constraints. These constraints specify that the newly-allocated region is
%separated from the domain of the old heap store $\hstore$ and the new
%domain of the new heap store include both of them. We use
%the notation $\uplus$ to succinctly represent these constraints.

%It is important to note here that KLEE (and LLBMC) in fact does not use symbolic
%variables as defined above, but instead use \emph{concrete} address
%values, obtained by running {\tt malloc}.

Let $\symstate_{0} \equiv \langle \pcstart, true \rangle$ 
be the 
% distinguished 
\emph{initial} symbolic state. A \emph{symbolic path}
$\symstate_{0} \rightarrow \symstate_{1}
\cdots \rightarrow \symstate_{m}$ denoted by $\overrightarrow{s}$ 
is a sequence of symbolic states such
that $\forall 1 \leq i \leq m$,  $\symstate_{i}$
is a \emph{successor} of $\symstate_{i-1}$.
We can now define \emph{symbolic exploration\/} as the process of
constructing a \emph{Symbolic Execution Tree} (SET) rooted at
$\symstate_{0}$.  Following some
\emph{search strategy}, the order in which the nodes are constructed 
can be different. For bounded verification and testing,
we assume that the tree depth is bounded.
Note that only one of our program statements, $assume(e)$, is a ``branch''.
This means that in a SET, a state has at most two successors.

The reachability of an \texttt{error} statement indicates a \emph{bug}.
Symbolic execution typically stops the path and generates
a \emph{failed} test case witnessing that bug. On the other hand, a
path \emph{safely} terminates if we reach a \texttt{halt} statement, 
and we also generate a \emph{passed} test case.  We prove a program
is \emph{safe} by showing that no \texttt{error} statement is reached.
A subtree is called \emph{safe} if no \texttt{error} statement is reached 
from its root.

%A path
%is \emph{safe} if it terminates at a program point where there is no
%transition emanating from or if it is infeasible.

%% file: interpolation.tex
We now present a formulation of the method of dynamic symbolic execution
with interpolation (DSEI) \cite{jaffar09intp,mcmillan10lazy}.
The essential idea, in brief, is this.
In exploring the SET, an \emph{interpolant} $\Intpsymbol$ of a state
$\symstate$ is an \emph{abstraction} of it which ensures the
safety of the subtree rooted at that state.  In other words, if we
continue the execution with $\Intpsymbol$ instead of $\symstate$,
we will \emph{not} reach any error.  Thus upon
encountering a state $\overline{\symstate}$ of the same
program point as $\symstate$, i.e., $\symstate$ and $\overline{\symstate}$ 
have the same set of emanating transitions,
if $\overline{\symstate} \models \Intpsymbol$, then continuing the
execution from $\overline{\symstate}$ will not lead to an error.
Consequently, we can prune the subtree rooted at $\overline{\symstate}$.

DSEI is a \emph{top-down} method, because it traverses the
SET from the root, and is also \emph{bottom-up} because it propagates
formulas backward from the end of a path in the SET.  Before
proceeding with describing DSEI, let us briefly argue that using just
a bottom-up approach using the concept of \emph{weakest precondition},
is not practical.  Before proceeding, consider using classic weakest
preconditions, which operate on program fragments, not paths.  The main
disadvantage here is that being entirely bottom-up,
the computed precondition at a program point
is \emph{agnostic} to the \emph{context} of the states which reach
that program point.

In contrast, DSEI performs a top-down depth-first search
of the state space, path by path. For each path,
it computes a \emph{path interpolant}.
This is where DSEI is bottom-up.
For each subtree, it computes
a \emph{tree interpolant} being the conjunction of all of the path
interpolants from within the subtree.

Now, we describe how to compute a path interpolant.  This is computed
recursively in a bottom-up manner.  Consider $\overrightarrow{s}$
being a sequence of states, $\overrightarrow{s}.\stmt$ to denote the
sequence $\overrightarrow{s}$ appended with a state $s$ where $s$ is
obtained by executing $\stmt{}$ from the last state in
$\overrightarrow{s}$, and the postcondition $post$, the function {\sc
intp()} computes the path-based weakest precondition interpolant.  In
case $\overrightarrow{s}$ is the empty sequence, denoted by
$\epsilon$, then $\overrightarrow{s}.\stmt$ is simply the state
obtained by executing $\stmt{}$ on the initial state.

\vspace*{-1mm}
\begin{definition}[Path-based Weakest Precondition] \label{def:path-wp}
Path-based weakest precondition is the weakest precondition of a symbolic path  $\overrightarrow{s}$ with respect to a postcondition $Post$.
This is formalized in Fig. \ref{algm:path-based}.
$\Box$
\end{definition}
\vspace*{-1mm}

\vspace*{-2mm}
%\hspace*{2mm}
\begin{figure}[!htb]
	\fbox{
		\small
		$		
		\begin{array}{lll}
		(1a)&\interpolant(\epsilon, \Psi) &= \Psi\\
		(2a)&\interpolant(\overrightarrow{s}.{\color{blue} x = e}, \Psi)&= \interpolant(\overrightarrow{s}, \Psi[x/e])    \\
		(3a)&\interpolant(\overrightarrow{s}.{\color{blue} assume(e)}, f\!alse)&= \interpolant(\overrightarrow{s}, \neg e )    \\
		(4a)&\interpolant(\overrightarrow{s}.{\color{blue} assume(e)}, \Psi)&= \interpolant(\overrightarrow{s}, e \implies\Psi)    \\
		\end{array}
		$
	}
	\vspace*{-2mm}
	\caption{ Path-Based Weakest Precondition}\label{algm:path-based}
	\vspace*{-3mm}
\end{figure}

See Fig. \ref{algm:path-based}. Rule (1a) and (2a) are the base cases
and their process is well-understood. We simply present an example.
 Suppose $x = e$ were {\tt x = x + 5} and $\Psi$ were {\tt x < 7}, then
the weakest precondition is {\tt x < 2}.  In general, for an
assignment of a variable to an expression $e$, say {\tt x
	= x + $e$}, the weakest precondition wrt. to $\Psi$ is
$\Psi[x/x+e]$, i.e. the formula obtained from $\Psi$ by
simultaneously replacing all its occurrences of $x$ with $x + e$.

Next, rule (3a) addresses the case that a node is infeasible in the SET. Finally, 
 we should highlight that in rule (4a), the interpolant $e \implies\Psi$ still 
 generates a \emph{disjunction}. Consequently, the path-based weakest precondition,
in general, may be a very large \emph{disjunction}, exponential in the program length. 

%Finally, we note that the \emph{base cases}, i.e. to define what an interpolant is at a terminal node in the SET would be:

%\vspace*{-2mm}
%\[
%\begin{array}{lll}
%Sa\!f\!e~terminal~node &\interpolant(\overrightarrow{s}, Post) & i\!f~ \overrightarrow{s}~\models~Post \\
%CounterExample~f\!ound  &\interpolant(\overrightarrow{s}, f\!alse) & i\!f~\overrightarrow{s}~\not\models~Post \\
%In\!f\!easible~node &\interpolant(\overrightarrow{s}, \neg e) & i\!f~s~\not\models~e \\
%\end{array}
%\]

\begin{example}
Consider the example in Fig. \ref{fig:example},
where {\tt $b[i]$} is a symbolic bit-vector , $Pre$ an
unspecified Boolean condition on the bit-vector representing
the \emph{precondition}, and $Post$ an
unspecified Boolean condition on any variables,
representing the \emph{postcondition}.
%The weakest precondition at point \pp{N}, call this $W_N$,  is:
%
%\vspace{-3mm}
%\begin{equation}
%ite(b[N], \phi [x[i] / 1],  \allowbreak   \phi [x[i] / -1])\footnote{
%	$ite(a, b, c)$ is short for $a \rightarrow b \wedge \neg a \rightarrow c$. \\$\phi[a/b]$ means simultaneously replacing all occurrences of $a$ in $\phi$ with $b$. }
%\end{equation}
%
%\noindent
%Similarly, $W_{N-1}$ is:
%
%\vspace{-3mm}
%\begin{equation}
%ite(b[N-1], W_N [x[N-1] / 1], W_N [x[N-1] / -1])
%\end{equation}
%
%\noindent
%and finally, $W_1$ is:
%
%\vspace{-3mm}
%\begin{equation}
%ite(b[1], W_2 [X[1] / 1], W_2 [X[1] / -1])
%\end{equation}
%
%\noindent
%Note that this formula contains only variables $b[i]$.
%Now, proving safety is tantamount to proving
%$\Psi \models W_1$, an arbitrary boolean formula.
%Thus in general this proof is intractable.

\vspace*{-3mm}
\begin{figure}[!htb]

\begin{lstlisting}[mathescape]
bool b[N]; 
assume($Pre$); 
$\pp{1}$ if (b[1]) $\kappa_1$ = 1; else $\kappa_1$ = -1; 
$\pp{2}$ if (b[2]) $\kappa_2$ = 1; else $\kappa_2$ = -1; 
$\cdots$ 
$\pp{N}$ if (b[N]) $\kappa_n$ = 1; else $\kappa_n$ = -1; 
assert($Post$); 
\end{lstlisting} 
% assert(-N <= x[1] + x[2] + $\cdots$ + x[N] <= N);
\vspace{-5mm}
\caption{Example}\label{fig:example}
\vspace{-3mm}
\end{figure}
%

%We return to our running example in Figure \ref{fig:running}.
Suppose $N = 3$, and $Pre$ implies $b[1] == b[2] \wedge b[2] == b[3]$.

The classic weakest precondition of the program with postcondition
$Post$ is a disjunction of 8 distinct formulas, each of which is (a)
an assignment of specific bit values to the variables $b[i]$ conjoined
with (b) a formula representing the propagation of the postcondition
$Post$ through the $\kappa_i$ assignments corresponding to the $b[i]$
assignments in (a).  These 8 formulas are

\smallskip 
$\begin{array}{l}
b[1] == 1 \wedge b[2] == 1 \wedge b[3] == 1 \wedge  Post [ \kappa_1 /1, \kappa_2 / 1, \kappa_3 / 1] \\
b[1] == 1 \wedge b[2] == 1 \wedge b[3] == 0 \wedge  Post [ \kappa_1 /1, \kappa_2 / 1, \kappa_3 / -1] \\
b[1] == 1 \wedge b[2] == 0 \wedge b[3] == 1 \wedge Post [ \kappa_1 / 1, \kappa_2 /-1, \kappa_3 / 1] \\
b[1] == 1 \wedge b[2] == 0 \wedge b[3] == 0 \wedge  Post [ \kappa_1 /1, \kappa_2 / -1, \kappa_3 / -1] \\
b[1] == 0 \wedge b[2] == 1 \wedge b[3] == 1 \wedge  Post [ \kappa_1 /-1, \kappa_2 / 1, \kappa_3 / 1] \\
b[1] == 0 \wedge b[2] == 1 \wedge b[3] == 0 \wedge  Post [ \kappa_1 /-1, \kappa_2 / 1, \kappa_3 / -1] \\
b[1] == 0 \wedge b[2] == 0 \wedge b[3] == 1 \wedge Post [ \kappa_1 / -1, \kappa_2 /-1, \kappa_3 / 1] \\
b[1] == 0 \wedge b[2] == 0 \wedge b[3] == 0 \wedge  Post [ \kappa_1 /-1, \kappa_2 / -1, \kappa_3 / -1] \\
\end{array}$
\smallskip

\noindent
By way of comparison, in the \emph{path-based} weakest precondition,
6 of 8 of these formulas would be unsatisfiable.
Thus, the \emph{path-based} weakest precondition,
after simplification, 
would be the following 2 formulas at \pp{1}. 

\vspace*{1mm}
$b[1] == 1 \wedge b[2] == 1 \wedge b[3] == 1 \wedge  Post [ \kappa_1 / 1, \kappa_2 / 1, \kappa_3 / 1]$ \\
\indent $b[1] == 0 \wedge b[2] == 0 \wedge b[3] == 0 \wedge  Post [ \kappa_1 /\!-\!\!1, \kappa_2 / \!-\!\!1, \kappa_3 / \!-\!\!1]$.
\vspace*{1mm} $\Box$
\end{example}

\noindent
The idea here is that classic weakest precondition, as it traverses is agnostic to precondition,
and there it essentially traverses every path through the program.
Only at the end it is known that only 2 of the 8 formulas generated cover the precondition.

We finally comment that in the path-based weakest precondition,
though it has only two formulas, this still is a \emph{disjunction}.
In the next two sections, we present our  algorithm which
\emph{approximates} the weakest path-based weakest precondition
with a single \emph{conjunction}.

%% file: algm.tex
%First, we begin with an example.

\begin{figure}
	% \centering
	\fbox{
		\begin{minipage}{102mm}
			
			\ \\
			%\begin{itemize}
			%\item
			\vspace*{-5mm}
			
			{\sc DSEI}(s) \{ // returns a tree interpolant $(\Intpsymbol)$ \\
			\pp{1}\hspace*{1mm} \textbf{if} ({\sc Unsat}($\eval{s}{}{}$)) \textbf{return} \false; ~~~ // {\color{blue} Infeasible node}\\  
			\pp{2}\hspace*{1mm} \textbf{if} (s $\models$ \Intpsymbol) \textbf{return} \Intpsymbol; \\
			\pp{3}\hspace*{1mm} \textbf{if} (s is terminal) \{ \\
			\pp{4}\hspace*{6mm} \textbf{if} ($s \models \Phi$)
			return $\Phi$; ~~~ // {\color{blue} $\Phi$ is the safety property}\\
			\pp{5}\hspace*{6mm} \textbf{else} return error;~~~ // {\color{blue} Counter-example}\\ 
			\pp{6}\hspace*{1mm} \} \\
			\pp{7}\hspace*{1mm} $\Intpsymbol'$ = $\Intpsymbol''$ = \true; \\
			\pp{8}\hspace*{1mm} \textbf{if} there is a transition $s \xrightarrow{\sf stmt} s'$ \{ \\
			\pp{9} \hspace*{6mm} $\Intpsymbol$ = {\sc DSEI}($s'$); \\
			\pp{10} \hspace*{6mm} \textbf{if} ($\Intpsymbol \not\equiv \false$ and $\Intpsymbol \not\equiv$ error) $\Intpsymbol'$ = \backprop($s$, stmt, $\Intpsymbol$); \\
			\pp{11} \hspace*{6mm} \textbf{else} return $\Intpsymbol$; \\
			\pp{12} \hspace*{0mm} \} \\
			\pp{13} \hspace*{0mm} \textbf{if} there is another transition $s \xrightarrow{\sf stmt'} s''$ \{ \\
			\pp{14} \hspace*{6mm} $\Intpsymbol$ = {\sc DSEI}($s''$); \\
			\pp{15} \hspace*{6mm} \textbf{if} ($\Intpsymbol \not\equiv \false$ and $\Intpsymbol \not\equiv$ error) $\Intpsymbol''$ = \backprop($s$, $stmt', \Intpsymbol$); \\
			\pp{16} \hspace*{6mm} \textbf{else} return $\Intpsymbol$; \\
			\pp{17} \hspace*{1mm} \} \\
			\pp{18} \hspace*{1mm} \textbf{return} $\Intpsymbol' \wedge \Intpsymbol''$;\\
			\pp{19}  			\}
			\vspace*{-2mm}
			\caption{The DSEI Algorithm}
			\label{algm1}
		\end{minipage}
	}
	\vspace*{-6mm}
\end{figure}
%\vspace*{-2mm}

The overall structure of our idealized algorithm is in Fig. \ref{algm1}, 
where it processes a symbolic execution tree (SET)  via DSEI. 
The function {\sc DSEI} receives a symbolic state $s$ which represents 
a node from the SET. It then processes the subtree beneath the node
 and returns an interpolant $\Intpsymbol$.
If $s$ is safe, i.e. all paths from $s$ do not violate the assertion, then 
$\Intpsymbol$ is an \emph{interpolant} storing  a generalization of the 
state $s$. If, on the other hand, $s$ is unsafe, i.e. there is one path from $s$
that violates assertion, then $\Intpsymbol$ is \texttt{error}.

The {\sc DSEI}$(s)$ function first
checks if a state is infeasible (line \pp{1}). The
return value is simply $\false$ (as the interpolant).
 
At lines \pp{3}-\pp{6} (when $s$ is a terminal node, i.e. it has no successors), it is checked if
$s$ is safe, then $\Phi$ is returned as an interpolant for
the \emph{safe} terminal node. Otherwise, a
counter-example is found and {\tt error} is returned.

Next, in lines \pp{8}-\pp{12}, the {\sc DSEI} function
is recursively called on $s'$ which is a successor node of $s$ via the
transition $s \xrightarrow{\sf stmt} s'$.  After the call is returned
(line \pp{10}), it is checked if the interpolant $\Intpsymbol$ is
$f\!alse$ or {\tt error}. If so, $f\!alse$ or {\tt error} is returned (line \pp{11}).

The key element of the algorithm is in line \pp{10}, where
the \backprop(s, stmt, $\Intpsymbol$) function returns $\Intpsymbol'$, a
practical estimation of the \emph{weakest precondition} of the
statement $stmt$ wrt. a postcondition $\Intpsymbol$. Note importantly
that here we cannot use the \emph{path-based} weakest precondition
from Section \ref{sec:statesub}, since it generates a disjunctive
interpolant.  In Section \ref{sec:algm1}, we present
the \emph{interpolation algorithm} of the \backprop{} function.

Similarly, in lines \pp{13}-\pp{17}, if another
transition $s \xrightarrow{\sf stmt'} s''$ exists,
the successor node $s''$ is processed and the 
interpolant $\Intpsymbol''$ is computed. 

Finally, the \emph{tree interpolant} is returned in line \pp{20} which 
is the conjunction of $\Intpsymbol'$ and
$\Intpsymbol''$. Note here that $\Intpsymbol' \wedge \Intpsymbol''$
contains a condition which while not violated, ensures
all paths within the subtree beneath $s$ will not lead to a violation
of the safety property $\Phi$.

The key to performance is the subsumption step in line \pp{2}.
Here $\Intpsymbol$ contains the interpolant generated from 
processing a similar node in the SET,
i.e. a node with the same program point. It is checked if
$s \models \Intpsymbol$. If so, $s$ is subsumed by $\Intpsymbol$
meaning all paths beneath the node $s$ would not lead to a violation
of the safety property $\Phi$.

This, in turn, means that the key is the quality of the interpolant $\Intpsymbol$ generated.
Without interpolation, i.e. relying on state subsumption alone, is not good
enough\footnote{This is because different symbolic states, that share the same
	program point, might be significantly different.}.  
In Section \ref{sec:algm1}, we present the \emph{interpolation algorithm} of \backprop{}  
which is the main technical contribution of this paper. 
The \backprop{}  function  computes a \emph{conjunctive approximation} of the 
\emph{weakest precondition} as an interpolant.

\vspace{2mm}
\noindent
\textbf{Remark on Depth First vs. Random Strategy:}
The algorithm presented in this section is based on the depth-first traversal (DFS) of the SET. 
However, if the SET is so big that full coverage is implausible,
then a DFS strategy in a non-pruning DSE (like KLEE) is 
known to have poor coverage \cite{cadar2008exe}.
For such large programs, a \emph{random} search strategy can be used 
to avoid the exploration of the SET from being stuck in some part of the SET. 
KLEE's ``random'' strategy presented in \cite{cadar08klee} 
addresses this issue and attempts to  maximize coverage.

We can utilize a random strategy in our 
algorithm too. This strategy would be similar  
to KLEE's random strategy \cite{cadar08klee}.
The difference between our random strategy and KLEE's random strategy
is two-fold: 1) Our approach still attempts to generate path and 
tree-interpolants to prune the SET. 2) Our random strategy adds 
one more step to the two steps in KLEE's random strategy, which are executed  in
 a round-robin fashion. This new step will pick a state which has a 
 high chance to create a tree interpolant.

For our algorithm, the choice between DFS and random 
search strategies can matter greatly.  DFS strategy is ideal 
for our algorithm since it maximizes the generation of tree-interpolants,
which in turn increases the chance of subsumption. 
In general, in the random strategy,
 tree interpolant are formed slower, and memory usage is  higher.
On the other hand, it can reach higher coverage when the SET is not fully traversed.
We follow the commonsense belief that for full exploration
DFS is as good as any other strategy, whereas, for incomplete exploration,
random is better.
In Section \ref{sec:imp_and_eval}, we will experiment with 
both the DFS and random strategies for our algorithm and 
we will show that both strategies can reach higher 
performance on  different programs.

%% file: algm1.tex
The essential idea behind our new interpolation algorithm is to 
have it be a \emph{conjunction}.
From section \ref{sec:statesub}, we have seen that the path-based weakest precondition is in general a disjunction, hence, still not practical. 
So our proposed approach will entail
a \emph{conjunctive approximation} of the path-based weakest precondition.
Clearly, the context itself is a first candidate.
We now show that we can, by using the context as a guide,  compute
an effective abstraction of it.

Before we proceed, we mention abstraction learning via interpolation \cite{jaffar09intp, mcmillan10lazy,mcmillan14lazy} which has demonstrated significant speedup in verification and testing,
e.g., \cite{jaffar12tracer,jaffar13test}.
Although in all of these previous efforts, the interpolants implemented 
are \emph{conjunctive}, they were not as general as in the weakest precondition.
In fact, all implementations ensured that an interpolant was
in the form of a \emph{conjunction} which could then be dealt with efficiently
by an SMT solver. 

We now present the \backprop~ function in Fig. \ref{algm:conjunctive} which generates a \emph{conjunctive approximation} of 
the path-based weakest precondition.
 As before, we use the notation $\overrightarrow{s}$ to denote a sequence
of states.
%Finally, we write $\eval{s_m}{}{} \models \psi$
%where $s_m $ is the last state in $\overrightarrow{s}$.
%(In case $\overrightarrow{s}$ in $\epsilon$, the entailment is just true.)

\begin{figure}[]
\fbox{
\footnotesize
$	
\hspace{-2mm}
\begin{array}{lll}
(1b)~\backprop(s,\epsilon, \Psi) &\equiv \Psi & \\
(2b)~\backprop(s,{\color{blue} x = e}, \Psi)    &\equiv \Psi[x/e]   & \\
(3b)~\backprop(s,{\color{blue} assume(e)}, f\!als\!e)& \equiv  \neg e  \\
(4b)~\backprop(s,{\color{blue} assume(e)}, \Psi) &\equiv  \Psi \wedge e    & i\!f ~\eval{s}{}{} \models e\\
(5b)~\backprop(s,{\color{blue} assume(e)}, 
\Psi)&\equiv \abduction(\mapstatetoformula{s},e,\Psi) & other\!wise \\
\end{array}
$ 
}
\vspace{-2mm}
\caption{Conjunctive Path-Based Weakest Precondition}\label{algm:conjunctive}
\vspace{-4mm}
\end{figure}

Rules (1b) - (3b) are similar to the respective rules in Fig. \ref{algm:path-based}.
The problem at hand is to determine an interpolant which
(a) it can replace $e \implies \Psi$ from Fig. \ref{algm:path-based}, and (b) it is a \emph{conjunction}.
The major difference here compared to the
path-based weakest precondition (from Fig. \ref{algm:path-based}) is that the rule (4a) is now
replaced by the two new rules (4b) and (5b).

We first dispense the easy case in rule (4b) where
$~\eval{s}{}{} \models e$ holds.
Clearly, $\Psi \wedge e$ is the right interpolant.

Next, we discuss the difficult case in rule (5b).  We know that
the best interpolant is provided by the concept of weakest
precondition, that is: $\neg e \bigvee \Psi$.  However, this is a
disjunction, and so is not suitable for us.  What we require is a
general method, for generalizing the constraints, that attempt to be as
powerful as the weakest precondition method. However, it is restricted to
produce only a conjunction of constraints.

Note that in rule (5b), we invoke a function $\abduction(\eval{s}{}{}, e, \Psi)$.
% in order to compute $~\backprop(s,{\color{blue} assume(e)},\Psi)$.
We already have that $\eval{s}{}{} \wedge e \models \Psi$ holds.
This is by virtue of the top-down computation that brought us to this point.
Now, we want a generalization 
$\overline{\eval{s}{}{}}$ 
such that $\overline{\eval{s}{}{}} \wedge e \models \Psi$.
This means we have an \emph{abduction} problem:
given a conclusion $\Psi$ and a partial contribution $e$ to that conclusion,
what is the most general constraint needed to be added to $e$?
This is a classic problem \cite{wikiabduction}. 
Unfortunately, we are not aware of any general abduction algorithm
that is practical for our purposes.
Next, we present our own abduction algorithm.

\begin{figure}[]
	% \centering
	\fbox{
		\begin{minipage}{102mm}
			\ \\
			$\abduction(\phi, e, \Psi)$ \{ \hspace*{5mm} // {\color{blue}given $\phi \wedge e \models \Psi$} \\
			\hspace*{1mm}$\bar{\phi}$ = {\sc core}($\phi \wedge e,  \Psi$)  \\
			\hspace*{1mm} let $v$ be the variables in $e$ \\
			\hspace*{1mm}$<\!\phi_{v}, \phi_{\bar{v}}\!\!>$ = {\sc separate}($\bar{\phi}$,  $v$) \hspace*{18mm}// {\color{blue}$\phi_{\bar{v}} \star (\phi_{v} \wedge e)  $} \\
			\hspace*{1mm} let $v_2$ be the variables in $(\phi_{v} \wedge e)$ \\
			\hspace*{1mm}$<\!\Psi_{v}, \Psi_{\bar{v}}\!\!>$ = {\sc separate}($\Psi$, $v_2$) \hspace*{16mm}// {\color{blue}$\Psi_{\bar{v}} \star (\phi_{v} \wedge e) \wedge \Psi_{\bar{v}} \star \Psi_{v}$} \\
			\hspace*{1mm} \textbf{if} ($\Psi_{v} \equiv true$) \textbf{return} $\Bar{\Bar{\phi}} \equiv \Psi_{\bar{v}}$ \\
			\hspace*{1mm} \textbf{return} $\Bar{\Bar{\phi}} \equiv \phi_{v} \wedge \Psi_{\bar{v}}$; \\
			\}
			\\
			\\
			$\core(\gamma,\Psi)$ \{ \hspace*{10mm}// {\color{blue} Assume $\gamma$ is of the form  $C_1 \wedge \ldots \wedge C_{n}$}\\
			\hspace*{1mm} \textbf{for} $i = 1 .. n$ \textbf{do}\\ 
			\hspace*{6mm}  \textbf{if}~$\gamma - C_i \models \Psi$ \textbf{then} $\gamma \equiv \gamma - C_i$\\
			\hspace*{1mm} \textbf{return} $\gamma$; \\
			\}
			\\
			\\
			$\separate(\gamma, v)$ \{ \hspace*{7mm} // {\color{blue}Assume $\gamma$ is of the form $C_1 \wedge \ldots \wedge C_{n}$}\\
			\hspace*{5mm}  \textbf{loop} \{ \\
			\hspace*{5mm}  $oldv$ = $v$; \\
			\hspace*{5mm} \textbf{for} $i = 1 .. n$ \{ \\ 
			\hspace*{10mm}  \textbf{let} $v_i$ be the variables in $C_i$ \\
			\hspace*{10mm}  \textbf{if}~$v_i \cap v \neq \emptyset$ \textbf{then} \{\\
			\hspace*{15mm}  $v = v \cup v_i$ \\
			\hspace*{15mm}  $\gamma_v = \gamma_v \wedge C_i$ \\
			\hspace*{10mm}  \} \textbf{else} $\gamma_{\bar{v}} = \gamma_{\bar{v}} \wedge C_i$ \\
			\hspace*{5mm} \textbf{until}~ ($v$ == $oldv$);   \}  \hspace*{18mm}// {\color{blue}fixpoint condition}\\
			\hspace*{5mm} \textbf{return} $<\gamma_{v}, \gamma_{\bar{v}}\!\!>$; \\
			\}
			\vspace*{-4mm}
			\caption{The Abduction Algorithm}
			\label{algm:abduction}
		\end{minipage}
	}
	\vspace{-4mm}
\end{figure}

%\vspace{-2mm}
\subsection{The Abduction Algorithm}

%{\color{magenta}
%	\begin{itemize}
%		\item Partition $\phi$ into three: $\phi_1, \phi_2, \phi_3$. \\
%		Now compute any $\phi_1$ such that $\phi_2 \wedge \phi_3 \wedge e \models \Psi$ \\
%		The second component $\phi_2$ is such that it is SEPARATE from $e$. \\
%		\item
%		Partition $\Psi$ into two: $\Psi_1, \Psi_2$. \\
%		The second component $\Psi_2$ is SEPARATE from $e$.
%	\end{itemize}
%	
%	\bigskip
%	
%	Next: the definition of SEPARATE(F1, F2).
%	
%	\bigskip
%	
%	\begin{itemize}
%		\item First, var(F1) is disjoint from var(F2).
%		
%		\item Cartesian closed of the solution tuples of $F1 \& F2$,
%		the closure happening between var(F1) and var(F2) 
%	\end{itemize}
%	
%}

Consider the abduction algorithm in Fig. \ref{algm:abduction}.  It uses
two main functions.  The first is \core($\gamma$, $\Psi$).  The algorithm
for the \core~ function is presented in Fig. \ref{algm:abduction}. 
%
%Essentially, the \core~ function eliminates
%any constraint in $\gamma$ which is not needed for implying $\Psi$.
%First, the \core~ function 
It is called with the arguments $\phi \wedge e$ and $\Psi$.  Note that,
as explained in the previous section, we have $\phi \wedge e \models \Psi$.
The \core~ function eliminates the constraints in $\phi \wedge e$ that
are not needed for implying $\Psi$. The result is stored in $\bar{\phi}$.

The second function is \separate($\gamma, v$). 
% The algorithm for the \separate~ function is also presented in Fig. \ref{algm:abduction}. 
The function partitions $\gamma$ 
into two formulas $\gamma_{v}$ and  $\gamma_{\bar{v}}$. The algorithm
stores in $\gamma_{v}$, any constraint which is either containing a variable 
from $v$ or contains a variable which is appearing in another constraint in
$\gamma_{v}$. The remainder of the constraints are stored in $\gamma_{\bar{v}}$.

The important property that is required concerns a notion of \emph{separation}.
We first define this property in general.

\begin{definition}[Separation] \label{def:separation}

Consider a first-order formula $\Psi_1 \wedge \Psi_2$
where the variables of $\Psi_1$ and $\Psi_2$ are $v_1$ and $v_2$ respectively,
and $v_1$ and $v_2$ are disjoint.  Let $\Psi_1\!\!\mid_{v_1}$
denote the  projection of $\Psi$ onto $v_1$.
We say that $\Psi_1$ and $\Psi_2$ are \emph{separate}, written $\Psi_1 \star \Psi_2$, if
$\Psi_1 \wedge \Psi_2 \equiv \Psi_1\!\!\mid_{v1} \wedge \Psi_2\!\!\mid_{v2}$.
$\Box$

\ignore{
Consider two first-order formulas $\Psi_1$ and $\Psi_2$; their set of variables: $v_1$ and $v_2$; and their set of
solutions $Sol(\Psi_1)$ and $Sol(\Psi_2)$. 
We define $\Psi_1$ and $\Psi_2$  as separate (denoted by {\sc separate}($\Psi_1,\Psi_2$)), if, (1) there exist no shared variables between them: $v_1 \cap v_2 \equiv \emptyset$; and (2) the set of the solutions of $\Psi_1 \wedge \Psi_2$ over $v_1 \cup~v_2$
is equal to the Cartesian product of their solutions: $Sol(\Psi_1) \times Sol(\Psi_2)$. 
}

\end{definition}

We next explain the separation property that we require from the $\separate(\gamma, v)$
function, which returns $\gamma_{v}$ and  $\gamma_{\bar{v}}$.
This property is: $\gamma_v \star \gamma_{\bar{v}}$.

We now return to the abduction function.
Note that it calls \separate() twice.
In the first call, $\phi$ is partitioned into two: $\phi_v$ and $\phi_{\bar{v}}$. 
Since, \separate()~ is provided with the variables of $e$, the critical property that we would have is that 
the set of variables in $(\phi_v \wedge e)$ and $\phi_{\bar{v}}$ are disjoint.
Moreover $(\phi_v \wedge e) \star  \phi_{\bar{v}}$  also holds.

In the next call, $\Psi$ is partitioned into two. The difference 
here is that the \separate()~ function is provided with the variables of 
$(\phi_v \wedge e)$. The critical property that we would have is that 
the set of variables in $\Psi_v$ and $\Psi_{\bar{v}}$ are disjoint.
Moreover, the set of variables in $(\phi_v \wedge e)$  and $\Psi_{\bar{v}}$ 
are disjoint too. In the end, after the two calls to \separate(), the following holds.

\vspace*{3mm}
\hspace*{14mm} $(\phi_v \wedge e) \star \phi_{\bar{v}}$, $\Psi_v \star \Psi_{\bar{v}}$
$(\phi_v \wedge e) \star \Psi_{\bar{v}}$ and $\phi_v \star \Psi_{\bar{v}}$ \hfill(1)
\vspace*{3mm}

Finally, $\phi_v \wedge \Psi_{\bar{v}}$ is returned as a generalization
of $\phi$ such that $\Psi_{\bar{v}} \wedge \phi_v  \wedge e \models \Psi$
holds. The special case is when  $\Psi_{v}$ contains no 
constraints, i.e. it is $True$. In this case, $\Psi_{\bar{v}}$ contains
all of the constraints in $\Psi$. Since  $\Psi \wedge e \models \Psi$ holds 
obviously  $\Psi$ is returned as a generalization of $\phi$.

\noindent
We now outline a proof that the abduction algorithm is correct.
We first require some helper results.

\begin{corollary}[Frame Rule] \label{col:frame} 
Consider three first-order formulas $A$, $B$, and $C$. \\
Then $A \models B$ iff $C \star A \models C \star B$.
$\Box$
\end{corollary}

\begin{proof}
That $A \models B$ implies $C \star A \models C \star B$ is obvious.
To prove $C \star A \models C \star B$ implies $A \models B$,
we proceed by contradiction. Assume $A\theta$ is true while $B\theta$ is false
for some valuation $\theta$ on the variables of $A$ and $B$.
We can now extend $\theta$ to include the variables $C$ - call this evaluation $\theta'$.
So $(C \wedge A)\theta'$ implies $(C \wedge B)\theta'$.
In case $(C \wedge A)\theta'$ is true, $B\theta'$ must be true.  
This contradicts that assumption that $B\theta$ is false
because $\theta$ and $\theta'$ agree on the variables of $B$.
Similarly, in case $(C \wedge A)\theta'$ is false, then this contradicts the assumption
that $A\theta$ is true
because $\theta$ and $\theta'$ agree on the variables of $A$.
\end{proof}

\ignore{
\begin{corollary}[Frame Rule 3] \label{col:frame3} 
Consider four first-order formulas $A$, $B$, $C_1$ and $C_2$. \\
If $C_1 \models C_2$ and $C_1 \star A \models C_2 \star B$, then $A \models B$.
$\Box$
\end{corollary}

\begin{proof}
We have $C_1 \star A \models C_1 \star B$ from the Frame Rule.
We have $C_1 \star B \models C_2 \star B $ from the fact $C_1 \models C_2$.
Thus $C_1 \star A \models C_2 \star B$. 
\end{proof}
}

\begin{corollary}[Frame Rule 2] \label{col:frame2} 
Consider three first-order formulas $A$, $B$, $C$. \\
If $A \star B \models C$, and $A \star C$, 
then $B \models C$.
$\Box$
\end{corollary}

\begin{proof}
We have $A \star B \models C$. Since $A \star C$, we also have $A \star B \models A \star C$. 
By the Frame Rule, we have $B \models C$.
\end{proof}

%
%\bigskip
%Deletion
%
%\bigskip
%C -> W without guard g(x)
%
%\bigskip
%C -> Cx \& Cy, Cx is frame of x in C \\
%W -> Wx \& Wy, Wx is frame of x in W
%
%\begin{definition}[Frame]
%	Cx is a frame of x in C if C|x = Cx.  Cy is anti-frame.  Notation Cx * Cy
%\end{definition}
%
%
%Next prove if (C1 -> C2) \& (C1 * A -> C2 * B) then A -> B.
%
%\bigskip
%Now Cy * (Cx \& g) -> Wy * Wx implies Cx \& g -> Wx therefore Wy * Cx \& g -> Wy * Wx.

\vspace{-3mm}
\begin{theorem}
	The abduction algorithm in Fig. \ref{algm:abduction}, when
	input with the constraints $\phi$, $e$ and $\Psi$ where $\phi \wedge e \models \Psi$, 
    outputs $\Bar{\Bar{\phi}}$ a generalization over $\phi$ 
    such that 	$\phi \wedge e \models \Bar{\Bar{\phi}} \wedge e \models \Psi$.
	$\Box$
\end{theorem}
\vspace{-3mm}

\begin{proof}
Let $\Bar{\phi}$ be {\sc core}($\phi \wedge e,  \Psi$).  We prove the general case:

\vspace*{3mm}
\hspace*{-3mm}$\begin{array}{lll}
 & \phi \wedge e \models \Psi & \mbox{Given} \\
 \longrightarrow & \Bar{\phi}  \wedge e \models \Psi, & \mbox{Correctness of CORE()} \\
 
  \longrightarrow & \phi_v \wedge \phi_{\bar{v}} \wedge e \models \Psi_v \wedge \Psi_{\bar{v}}, & 
             <\!\phi_{v}, \phi_{\bar{v}}\!\!> = \separate(\bar{\phi},  v)~\mbox{and} <\!\Psi_{v}, \Psi_{\bar{v}}\!\!> = \separate(\Psi, v_2), \\
			& &\mbox{where} ~v~ \mbox{and} ~v_2~\mbox{are the set of variables in} ~e ~\mbox{and} ~(\phi_{v} \wedge e). \\
		
  \longrightarrow & \phi_{\bar{v}} ~\star{}~ (\phi_v \wedge e) \models \Psi_{\bar{v}} \star{} \Psi_v & 
            \mbox{by the separation property (1).} \\
            
     \longrightarrow & (a)~ \phi_{\bar{v}} ~\star{}~ (\phi_v \wedge e) \models \Psi_{\bar{v}}  &   \\
   & (b)~ \phi_{\bar{v}} ~\star{}~ (\phi_v \wedge e) \models  \Psi_v &  \\         
   
        \longrightarrow & (a')~  \phi_{\bar{v}} \models \Psi_{\bar{v}}  &  (\phi_v \wedge e) \star \Psi_{\bar{v}} \mbox{ and Frame Rule 2}   \\
   & (b')~  (\phi_v \wedge e)  \models  \Psi_{v} &  \mbox{by } \phi_{\bar{v}} \star \Psi_{v} \mbox{ and Frame Rule 2} \\
%  \longrightarrow & (a')~ \phi_{\bar{v}}  \models \Psi_{\bar{v}}  & \mbox{by Corollary \ref{col:frame2}}  \\
%     & (b'')~ (\phi_v \wedge e) \models \Psi_v &  \\
     
%  \longrightarrow & (a'')~\Psi_{\bar{v}}   \models \Psi_{\bar{v}}  &  \mbox{Obvious}\\
%     & (b')~   (\phi_v \wedge e) \models \Psi_v & \\
  
  \longrightarrow &  \Psi_{\bar{v}} \wedge (\phi_{v} \wedge e) \models  \Psi_{\bar{v}} \wedge \Psi_{v}, 
  &  \mbox{by } (\phi_v \wedge e) \star \Psi_{\bar{v}} \mbox{ and } \Psi_v \star \Psi_{\bar{v}}
  \mbox{ and the Frame Rule on}~ (b') \\
    \longrightarrow & \bar{\bar{\phi}} \wedge e \models \Psi & 
\end{array}$

%We now prove the special case where  $\Psi_{r} \equiv \Psi$ and $\Psi_{g} \equiv True$:
%
%\vspace*{3mm}
%\hspace*{-3mm}$\begin{array}{lll}
%& \phi \wedge e \models \Psi & \mbox{Given} \\
%\longrightarrow & \Bar{\phi}  \wedge e \models \Psi, & \mbox{Correctness of CORE()} \\
%\longrightarrow & \phi_g \wedge \phi_{r} \wedge e \models \Psi_g \wedge \Psi_{r}, & 
%\mbox{Splitting $\Bar{\phi}$ and $\Psi$  s.t. $\Psi_{g} \star  \Psi_r$ and $\phi_{g} \star  \phi_r$} \\
%\longrightarrow & (a)~ \phi_g \wedge e \models \Psi_g, & vars(\phi_g) \cap vars(\phi_{r}) = \emptyset \wedge \\
%& (b)~ \phi_{r}  \wedge e \models \Psi_{r}, & vars(\Psi_g) \cap vars(\Psi_{r}) = \emptyset \\
%\longrightarrow & (b)~ \phi_{r}  \wedge e \models \Psi_{r}, & \mbox{omit (a) since $\Psi_{g} \equiv True$} \\
%\longrightarrow &  \Psi_{r} \wedge e \models \Psi_{r}, &  \mbox{Obvious} \\
%\longrightarrow &  \Psi \wedge e \models \Psi, &  \mbox{$\Psi_{r} \equiv \Psi$} \\
%
%\longrightarrow & \bar{\bar{\phi}} \wedge e \models \Psi & 
%\end{array}$

\end{proof}

\begin{wrapfigure}[16]{r}{3.3in}
	\vspace{-3mm}
	%\begin{figure}[!htb]
	\begin{center}
		\vspace*{-3mm}
		\includegraphics[scale=0.65]{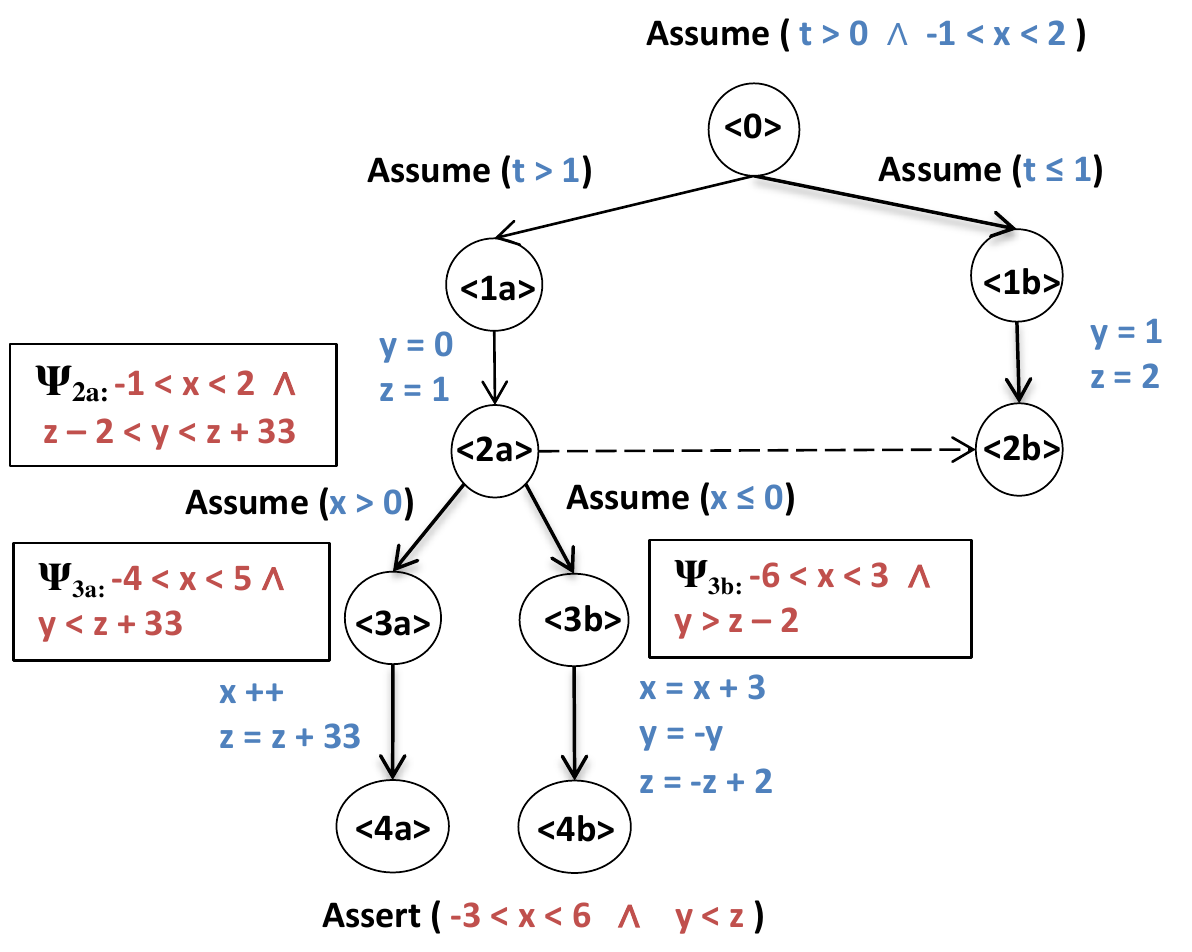}
		\vspace{-4mm}
		\caption{Example with Interpolation}\label{ex2}
	\end{center}

	%\end{figure}
\end{wrapfigure}

In summary, the rule (5b) from Fig \ref{algm:conjunctive} performs (a) removing a constraint that
already existed, or (b) transforming an existing constraint with
another existing constraint.  Step (a) is implemented via an
unsat-core method.  Step (b) on the other hand, provides a mechanism
for \emph{backward} reasoning, computing an approximation of the
weakest precondition.  That is, the interpolant is (partly) composed
of constraints that come from the postcondition in a \emph{bottom-up
manner}, and not from the constraint of the current symbolic state,
which is from a \emph{top-down manner}.

A naive implementation of the abduction algorithm will not be sufficiently practical.
We will describe how we implement the algorithm in Section  \ref{sec:impl-remark}.

We now exemplify the concepts behind the \backprop{} algorithm in the following synthetic  example.  
% 	The following example computes an abduction of $\omega$ wrt $assume(x > 0)$.
We attempt to show that first, the \backprop{} algorithm can generate a non-trivial interpolant;
and second that the interpolant
generated by the unsatisfiability core method is not an ideal solution, i.e. 
it is less general as compared to the \backprop{} algorithm.

\begin{example}\label{ex77}
	In Fig. \ref{ex2}, we depict the full SET of a program explored by DSEI.
	Note that program points are denoted by numbers,
	e.g. \pp{2}, and we attach small letters to distinguish different
	encounters of the same program point, e.g. \pp{2a}.
	Assume, $t$ and $x$ are symbolic variables and $y$ and $z$ are program variables. 
	We are attempting to prove the postcondition $-3<x<6 \wedge y<z$.
	
	The left most path is traversed to \pp{4a} and since it is safe, 
    an interpolant $\Psi_{3a}$ is generated at \pp{3a} using rule (2b): $-4 < x < 5~\wedge~y < z + 33$. 
	Moving now to the path \pp{1} \pp{2a} \pp{3b} \pp{4b}, the interpolant generated at \pp{3b}
	using rule (2b) would be $\Psi_{3b}$: $-6 < x < 3~\wedge~y > z - 2$.
%At \pp{2a}, this time the guard is $assume(x \leq 0)$,
%	and the postcondition is $\Psi_{3b}$ is $-6 < x < 3~\wedge~y > z - 2$.
%	We rule (5b) and the generated abduction formula is: $-1 < x < 2 \wedge y > z -2$.
%    Similarly, coming \pp{4b}, another interpolant $\Psi_{3b}$ is generated: 

    We now show how to propagate $\Psi_{3a}$ and $\Psi_{3b}$ to obtain $\Psi_{2a}$.
    Consider first $\Psi_{3a}$.
	At \pp{2a}, the guard in question is $assume(x > 0)$.
	% and the postcondition $\Psi_{3a}$ is $-4 < x < 5~\wedge~y < z + 33$.
	We apply rule (5b) and the \abduction{} function. First, 
	$\bar{\phi}$ would be $-1 < x < 2~\wedge~y = 0~\wedge~z = 1$,
	which still implies $\Psi_{3a}$. Next, 	
    since the guard has just the variable $x$,
	$\phi_{x}$ is  $-1 < x < 2$, and $\phi_{\bar{x}}$ is $y = 0 \wedge  z = 1$. 
    Finally, $\Psi_{x}$ would be $-4 < x < 5$ and $\Psi_{\bar{x}}$ is $y < z + 33$.
    In other words, the first formula $\phi_{x}$  is ``in the frame of $x$'', while
    $\phi_{\bar{x}}$ is in the ``anti-frame''.
	Thus the abduction formula $\Psi_{2a}$ is $-1 < x < 2 \wedge y < z + 33$, which is obtained
    as a combination of the incoming context from the top (which implies $-1 < x < 2$) and
    the formula $\Psi_{3a}$ which come from the bottom.
    Repeating this method for $\Psi_{3b}$, our abduction algorithm produces another 
    interpolant: $-1 < x < 2 \wedge y > z - 2$.
    Conjoining these two interpolants finally gives an interpolant $\Psi_{2a}$.
\ignore{	
	By intersecting the two abduction formulas, we obtain the \emph{tree interpolant} at \pp{2a}: $-1 < x < 2 \wedge z-2 < y < z + 33$. 
	Note that this is more general compared to the constraint store at \pp{2a} and replacing it with the constraint store would still maintain safety. 
}
	
	Now, moving to \pp{2b} the constraint store implies the generated interpolant
	from \pp{2a}. Hence, node \pp{2b} is subsumed (a.k.a. pruned)
	and its safety is inferred from the computed interpolant.
	Finally, we note that a classic Unsat-core interpolant would be $-1 < x < 2~\wedge~y = 0~\wedge~z = 1$ 
	which obviously would not be able to subsume node \pp{2b}.
	$\Box$
\end{example}

%\begin{table}
\begin{wraptable}[9]{r}{2.6in}
	\vspace{-4mm}
	\scriptsize
	\centering
	\begin{tabular}{|r|r|r|r|r|r|} 
		\hline
		\multicolumn{1}{|c|}{\multirow{2}{*}{\textbf{N}}} & \multicolumn{2}{c|}{\textbf{CBMC}} &
		\multicolumn{1}{c|}{\textbf{LLBMC}}  &
		\multicolumn{2}{c|}{\textbf{Our algorithm}} \\ 
		\cline{2-6}
		
		\multicolumn{1}{|c|}{}    &
		\multicolumn{1}{c|}{\textbf{Time}} &
		\multicolumn{1}{c|}{\textbf{Clauses}} & 
		\multicolumn{1}{c|}{\textbf{Time}}  & 
		\multicolumn{1}{c|}{\textbf{Time}} & 
		\multicolumn{1}{c|}{\textbf{SET Size}}  \\ \hline
		
		%10   & 0.02   & 5170      & 0.01    & 0.02   & 240 \\ 
		%\hline
		20  & 0.11   & 10549     & 0.02    & 0.04   & 470 \\ 
		\hline
		%50  & 6.21   & 26647     & 0.31    & 0.14   & 1160\\ 
		%\hline
		100 & 66.47  & 53509     & 2.66    & 0.42   & 2310\\ 
		\hline
		%200 & 520.95 & 107209    & 60.34   & 1.57   & 4610\\ 
		%\hline
		400 & 2460.98& 214609    & 750.41  & 7.83   & 9210\\ 
		\hline
		500 & \multicolumn{2}{c|}{$\infty$}  &
		1538.38 & 11.19  & 11510 \\ 
		\hline
		1000& \multicolumn{2}{c|}{$\infty$} &
		\multicolumn{1}{c|}{$\infty$}     & 53.85  & 23010        \\ 
		\hline
		2000& \multicolumn{2}{c|}{$\infty$} &
		\multicolumn{1}{c|}{$\infty$}     & 327.68 & 46010        \\
		\hline
	\end{tabular}
	%	\vspace{-4mm}
	\caption{Running Example in Fig. \ref{fig:example} on CBMC, LLBMC and Our algorithm, $\infty$
		indicates timeout (3600 seconds)}
	\label{table:motivEx}
\end{wraptable}
%\end{table}

We now reconsider the example in Fig. \ref{fig:example} and show
how the above interpolation algorithm
would deal with this example in a great way.
This time we do not constrain the precondition $Pre$ (i.e. the $b[i]$ can freely
take any binary values), and the postcondition $Post$ is 
$-N \leq x[1]+x[2]+\cdots+x[N] \leq N$ for various $N$.  
See Table \ref{table:motivEx} where we compare our algorithm 
 against CBMC \cite{cbmc} and
LLBMC \cite{llbmc}.  (KLEE only manages $N = 24$ within timeout.)
The reason for the vast superiority of our
algorithm is that at any level $i$ in the traversal, we compute 
just \emph{one} interpolant:

\vspace*{-4mm}
$$-N+i \leq x[1]+x[2]+\cdots+x[i] \leq N-i$$
\vspace*{-4mm}

\noindent
which subsumes all states at this level which are encountered later.
In other words, we have ``perfect'' subsumption.
Note that our search tree size is linear in $N$.

%% file: impl_remark.tex
We call our implementation \tracerx{}. It is implemented on top of
KLEE~\cite{cadar08klee}. The main addition to KLEE is the
implementation of interpolation. DSE with interpolation was
implemented before in \tracer{}
~\cite{jaffar11unbounded,jaffar12tracer}.  \tracerx{}
improves over \tracer{} by building on top of KLEE, and 
an enhanced interpolation algorithm which makes it more
efficient, and able to handle LLVM~\cite{llvm17}, including C/C++
programs.

\tracerx{} produces a new data structure called the \emph{subsumption table}.
This persistent structure is where the
interpolants, which contain the subset of the path condition as well
as the subset of memory regions are stored. When a new symbolic state
is encountered, this table is consulted to check if the new state is
subsumed by a record in the table, and hence its traversal need not
continue.  An entry in the table is created whenever KLEE removes a
state from its worklist, which we assume to mean that KLEE has
finished traversing the subtree originating from that state.

Now we explain how we  implement the \core{}() algorithm in Fig. \ref{algm:abduction}.
In Section \ref{sec:intro} we explained how an SMT solver 
can use the optimization method of (CDCL) 
\cite{marques1999grasp}.
More specifically, a core step in SMT solving is to involve a ``theory solver''
to solve a conjunction of constraints written for a particular theory.
The CDCL method requires that the solver not only decides the satisfiability
of the given conjunction.  In case the result is ``unsatisfiable'',
the solver also indicates which portion of the conjunction is required 
to keep in its unsatisfiability proof.  This is known as the \emph{unsatisfiability core}
of the conjunction.  In its pure form, the core only contains constraints that
were \emph{already encountered}. Essentially, we employ the \emph{unsatisfiability core}
technology of the SMT solvers  for a more efficient implementation of the \core{}() function. 

Next, we explain how we efficiently implement the \separate{}() algorithm in Fig. \ref{algm:abduction}.
We have employed a light-weight syntactic partitioning to approximate 
 the algorithm in the \separate{}() function.

Now, we will briefly explain how we extend the interpolation 
algorithm to other LLVM instructions. For this, we keep some extra information in the interpolant or the context.
First, we elaborate more on the operational semantics of the {\tt malloc} and {\tt free} instructions.
The key difference in how we deal with the malloc instruction, compared to KLEE,
 is that instead of using a concrete address returned by a system call to malloc, 
we use a fresh symbolic variable. We also add into the path condition the constraints
specifying that the newly-allocated region is separated from the domain of the old heap store 
and the new domain of the new heap store includes both of them. 

We also have special treatment for the array operation and the {\tt GEP} instruction.
 As an example, suppose the transition were  {\tt a[i] = 5}
and $\Psi$ was the formula  {\tt *p = 5}.
We have extended rule (2b) from Fig \ref{algm:conjunctive}
to return  {\tt <M, i, 5>[p] = 5} as the interpolant.
This  formula is to be understood in the array theory.
That is, {\tt M} is a distinguished array variable representing the
(entire) heap, {\tt <M, i, 5>} is an array expression representing
the array obtained from {\tt M} after the element 5 has been inserted
into location {\tt i}.  Finally, {\tt <M, i, 5>[p]} refers to
the $p^{th}$ element of this array expression.

In order to extend the interpolation algorithm to perform sound inter-procedural subsumption,
we store the call stack in an efficient way with an interpolant. This stored call stack is later checked 
with the call stack at the subsumption point and subsumption is only
allowed if the call stacks are identical.

% 99999

%% file: main_experiment.tex
We used an Intel Core i7-6700 at 3.40 GHz Linux box (Ubuntu 16.04)
with 32GB RAM. The programs in Tables \ref{tab:alltargets},
 \ref{tab:hardtargets}, \ref{tab:supplm1m2}, and  \ref{tab:supplblock-1} (47 programs) are from
SV-COMP Verification tasks \cite{psyco} and The Rigorous Examination
of Reactive Systems Challenge (RERS) \cite{rers}. A large subset of
the test programs are industrial programs or have
been used in testing and verification competitions.  The raw experimental results can be
accessed at \cite{Artifacts}\footnote{Due to size limitations, we were
not able to upload the raw experimental results as a supplementary material. 
Instead we uploaded the blinded experimental results at \cite{Artifacts}. 
Please note that the size after decompressing is > 5 GB.}.

\vspace{2mm}
%%%%%%%%%%%%%
\noindent
\textbf{Benchmarks Tested:}
Our first set (\benchmark{psyco1} to \benchmark{psyco7}) is from SV-COMP 
verification tasks \cite{psyco}. These programs are generated by the PSYCO 
tool \cite{psycotool} which produces interfaces using a symbolic execution 
and active automata learning.  These programs contain complicated loops
and are hard to analyze. 

The second set is from RERS\cite{rers} (prefixed with ``P'' and ``m'' in the
tables). They are from RERS Challenge competition in years 2012, 2017, 
and 2019 (identified by `-R12' to `-R19' respectively). The programs 
identified with `P' are from the three different categories of the 2012 
competition \cite{rers}: 
1) easy/small, containing plain assignments;
2) medium/moderate, containing arithmetic operations; and
3) large/hard, containing array and data structure manipulation.
 
The programs `P3-R17*', `P2-T-R17*', and `P11-R17*' are from the LTL and Reachability 
problems of RERS 2017\cite{rers17} respectively. These programs are from the small and 
moderate size group and easy to hard categories. 
Similarly the `P*-R19' problems are from the Sequential Training Problem, RERS 2019 \cite{rers19}.
The programs `m34*' and `m217*' are from Industrial Training Problems RERS 2019 
and are divided into LTL, CTL, and Reachability Training Problems. Since, we have 
tested LTL problems from other  tracks, here  we focused on CTL and Reachability groups. 
These programs were the most difficult and complex programs in our experiment.
We tagged the CTL and Reachability groups with `-C' and `-R', and the Arithmetic and 
Data Structure groups with `-A' and `-D'. Most of the programs are originally unbounded 
and we have tested them  with different bounds (the program name is suffixed by the bound
e.g. -100 means the loop bound used was 100).
%%%%%%%%%%%%%

%The programs in 
%Tables \ref{table:grp1}, \ref{table:grp2}, ($16 + 60 + 20 = 96$ programs) are from 
%the \textbf{GNU Coreutils benchmarks} (version 6.11) \cite{coreutils-6.11}.
We performed two experiments.

\begin{itemize}
\item[$\bullet$] The main experiment is on penetration/verification. \\
This experiment runs each program using 
\emph{one target at a time}.
We then considered a subset of the original targets called \emph{hard targets}.
These are obtained by filtering out targets
which can be proved easily by state-of-the-art methods: vanilla symbolic execution
for reachable targets, and static analysis for unreachable targets.
We then reran the main experiment on hard targets only.

\item[$\bullet$] The supplementary experiment is on testing/coverage. \\
It is modeled after the TEST-COMP competition which has a ``bug
finding'' component, and a ``coverage'' component.  In the first part,
bug finding, the task is to identify \emph{one} target among all the
targets injected in a program (performed for both all targets and hard
targets).

\hspace*{5mm}In the second part, the overall objective is to measure code coverage.
More precisely,
we measured the coverage of \emph{basic blocks}. Each program is ran with the purpose 
of \emph{full} exploration (timeout 1 hour), reporting any \emph{memory or assertion 
error} detected along the way.  (This is the default analysis of KLEE.)
We report the block coverage for the 47 programs  from SV-COMP and we also extend this 
 experiment to GNU Coreutils benchmarks \cite{coreutils-6.11}.
\end{itemize}

\noindent
In both experiments, our baselines are KLEE \cite{cadar08klee} and  CBMC \cite{cbmc}, 
as the state-of-the-art DSE and SSE tools. 
In general, CBMC is not appropriate for the second experiment on coverage
(because they react with an external environment).
Hence, there we only compared with KLEE, and use the Coreutils benchmark.

\begin{table*}[!htb]
\caption{The results for Main Experiment (All Targets)}
\label{tab:alltargets}
\footnotesize
	\begin{tabular}{|l|r|r|r|r|r|r|r|r|r|r|r|r|r|r|}
		\hline
	\multirow{3}{*}{Benchmark} & \multirow{3}{*}{\#AT} & \multicolumn{3}{c|}{KLEE}                             & \multicolumn{3}{c|}{CBMC}                             & \multicolumn{3}{c|}{TracerX}                        & \multirow{3}{*}{\#W} & \multirow{3}{*}{\#L} & Speedup & Speedup \\ \cline{3-11}
	
	&                       & Time        & \multirow{2}{*}{U} & \multirow{2}{*}{R} & Time        & \multirow{2}{*}{U} & \multirow{2}{*}{R} & Time      & \multirow{2}{*}{U} & \multirow{2}{*}{R} &                      &                       &        vs             &    vs                 \\ 
	
	&                       & (min)       &                    &                    & (min)       &                    &                    & (min)     &                    &                    &                      &                       &    KLEE                 &      CBMC               \\ \hline
		psyco1-100                 & 35                    & 153         & 0                  & 5                  & 175         & 0                  & 0                  & 2.5       & 30                 & 5                  & 30                   & 0                     & \leavevmode\color{red}{0.1}                 & -                   \\ \hline
	psyco1-500                 & 35                    & 152         & 0                  & 5                  & 175         & 0                  & 0                  & 48        & 30                 & 5                  & 30                   & 0                     & \leavevmode\color{red}{0.0}                 & -                   \\ \hline
	psyco1-1000                & 35                    & 152         & 0                  & 5                  & 175         & 0                  & 0                  & 151       & 0                  & 5                  & 0                    & 0                     & \leavevmode\color{red}{0.0}                 & -                   \\ \hline
	psyco2-8                   & 61                    & 285         & 0                  & 4                  & 128         & 57                 & 4                  & 17        & 57                 & 4                  & 0                    & 0                     & 1.4                 & 251                 \\ \hline
	psyco2-10                  & 61                    & 285         & 0                  & 4                  & 243         & 57                 & 4                  & 107       & 57                 & 4                  & 0                    & 0                     & 1.3                 & 454                 \\ \hline
	psyco2-12                  & 61                    & 285         & 0                  & 4                  & 305         & 0                  & 0                  & 287       & 0                  & 4                  & 0                    & 0                     & 1.0                 & -                   \\ \hline
	psyco3-8                   & 61                    & 285         & 0                  & 4                  & 124         & 57                 & 4                  & 17        & 57                 & 4                  & 0                    & 0                     & 1.2                 & 248                 \\ \hline
	psyco3-10                  & 61                    & 285         & 0                  & 4                  & 234         & 57                 & 4                  & 98        & 57                 & 4                  & 0                    & 0                     & 1.1                 & 447                 \\ \hline
	psyco3-12                  & 61                    & 285         & 0                  & 4                  & 305         & 0                  & 0                  & 287       & 0                  & 4                  & 0                    & 0                     & 1.0                 & -                   \\ \hline
	psyco4-8                   & 61                    & 285         & 0                  & 4                  & 109         & 57                 & 4                  & 15        & 57                 & 4                  & 0                    & 0                     & 1.4                 & 231                 \\ \hline
	psyco4-10                  & 61                    & 285         & 0                  & 4                  & 245         & 57                 & 4                  & 93        & 57                 & 4                  & 0                    & 0                     & 1.5                 & 474                 \\ \hline
	psyco4-12                  & 61                    & 285         & 0                  & 4                  & 305         & 0                  & 0                  & 287       & 0                  & 4                  & 0                    & 0                     & 1.2                 & -                   \\ \hline
	psyco5-100                 & 33                    & 141         & 0                  & 5                  & 165         & 0                  & 0                  & 2.6       & 28                 & 5                  & 28                   & 0                     & \leavevmode\color{red}{0.7}                 & -                   \\ \hline
	psyco5-500                 & 33                    & 142         & 0                  & 5                  & 165         & 0                  & 0                  & 50        & 28                 & 5                  & 28                   & 0                     & \leavevmode\color{red}{0.1}                 & -                   \\ \hline
	psyco5-1000                & 33                    & 142         & 0                  & 5                  & 165         & 0                  & 0                  & 141       & 0                  & 5                  & 0                    & 0                     & \leavevmode\color{red}{0.0}                 & -                   \\ \hline
	psyco6-100                 & 47                    & 154         & 0                  & 17                 & 235         & 0                  & 0                  & 4.5       & 30                 & 17                 & 30                   & 0                     & \leavevmode\color{red}{0.3}                 & -                   \\ \hline
	psyco6-500                 & 47                    & 156         & 0                  & 17                 & 235         & 0                  & 0                  & 90        & 30                 & 17                 & 30                   & 0                     & \leavevmode\color{red}{0.0}                 & -                   \\ \hline
	psyco6-1000                & 47                    & 154         & 0                  & 17                 & 235         & 0                  & 0                  & 153       & 0                  & 17                 & 0                    & 0                     & \leavevmode\color{red}{0.0}                 & -                   \\ \hline
	psyco7-8                   & 74                    & 381         & 0                  & 0                  & 370         & 0                  & 0                  & 10        & 74                 & 0                  & 74                   & 0                     & -                   & -                   \\ \hline
	psyco7-10                  & 74                    & 383         & 0                  & 0                  & 370         & 0                  & 0                  & 13        & 74                 & 0                  & 74                   & 0                     & -                   & -                   \\ \hline
	psyco7-12                  & 74                    & 381         & 0                  & 0                  & 370         & 0                  & 0                  & 17        & 74                 & 0                  & 74                   & 0                     & -                   & -                   \\ \hline
	m34-C-A-6                  & 213                   & 1013        & 0                  & 13                 & 725         & 133                & 0                  & 15        & 158                & 55                 & 67                   & 0                     & 188                 & 31                  \\ \hline
	m34-C-D-10                 & 712                   & 3540        & 0                  & 6                  & 3560        & 0                  & 0                  & 3162      & 0                  & 102                & 96                   & 0                     & 369                 & -                   \\ \hline
	m217-R-A-25                & 232                   & 1006        & 0                  & 54                 & 1155        & 3                  & 0                  & 826       & 0                  & 90                 & 36                   & \leavevmode\color{red}{3}                     & 35                  & -                   \\ \hline
	m217-R-D-3                & 100                   & 443         & 0                  & 12                 & 500         & 0                  & 0                  & 12        & 84                 & 16                 & 88                   & 0                     & 11                  & -                   \\ \hline
	P4-R12-8                   & 61                    & 300         & 0                  & 1                  & 305         & 0                  & 0                  & 36        & 60                 & 1                  & 60                   & 0                     & 1.1                 & -                   \\ \hline
	P5-R12-8                   & 61                    & 300         & 0                  & 1                  & 305         & 0                  & 0                  & 144       & 36                 & 25                 & 60                   & 0                     & \leavevmode\color{red}{0.5}                 & -                   \\ \hline
	P6-R12-8                   & 61                    & 293         & 0                  & 3                  & 266         & 35                 & 26                 & 99        & 35                 & 26                 & 0                    & 0                     & 1.8                 & 15                  \\ \hline
	P11-R19-50                 & 190                   & 436         & 0                  & 114                & 30          & 76                 & 114                & 398       & 0                  & 114                & 0                    & \leavevmode\color{red}{76}                    & 9.3                 & 13                  \\ \hline
	P12-R19-10                 & 101                   & 287         & 0                  & 46                 & 467         & 8                  & 0                  & 174       & 0                  & 67                 & 21                   & \leavevmode\color{red}{8}                     & 24                  & -                   \\ \hline
	P17-R12-8                  & 61                    & 184         & 0                  & 25                 & 305         & 0                  & 0                  & 183       & 0                  & 25                 & 0                    & 0                     & 2.2                 & -                   \\ \hline
	P2-T-R17-12                & 86                    & 192         & 0                  & 51                 & 355         & 23                 & 0                  & 3.9       & 35                 & 51                 & 12                   & 0                     & 94                  & 16                  \\ \hline
	P3-R17-7                   & 164                   & 492         & 0                  & 79                 & 820         & 0                  & 0                  & 9.9       & 85                 & 79                 & 85                   & 0                     & 60                  & -                   \\ \hline
	P11-R17-7                  & 163                   & 503         & 0                  & 68                 & 423         & 91                 & 0                  & 7.8       & 95                 & 68                 & 4                    & 0                     & -                   & 9.5                 \\ \hline
	P14-R12-20                 & 154                   & 468         & 0                  & 70                 & 770         & 0                  & 0                  & 37        & 35                 & 119                & 84                   & 0                     & 13                  & -                   \\ \hline
	m34-R-A-6                  & 413                   & 1779        & 0                  & 91                 & 1665        & 231                & 0                  & 339       & 322                & 91                 & 91                   & 0                     & 13                  & 3.5                 \\ \hline
	m217-C-A-10                & 32                    & 138         & 1                  & 4                  & 0.3         & 0                  & 32                 & 3.6       & 0                  & 32                 & 0                    & 0                     & 228                 & 1.6                 \\ \hline
	m217-C-D-10                & 98                    & 476         & 0                  & 3                  & 2.7         & 64                 & 34                 & 36        & 64                 & 34                 & 0                    & 0                     & 17                  & 2.2                 \\ \hline
	P1-R19-20                  & 30                    & 96          & 0                  & 12                 & 0.3         & 0                  & 30                 & 49        & 0                  & 21                 & 0                    & \leavevmode\color{red}{9}                     & 74                  & 3.7                 \\ \hline
	P2-R19-20                  & 38                    & 85          & 0                  & 24                 & 187         & 0                  & 2                  & 5.8       & 0                  & 37                 & 12                   & 0                     & 170                 & 325                 \\ \hline
	P3-R19-15                  & 75                    & 196         & 0                  & 40                 & 352         & 16                 & 0                  & 177       & 0                  & 40                 & 0                    & \leavevmode\color{red}{16}                    & 80                  & -                   \\ \hline
	P3-R12-8                   & 166                   & 454         & 0                  & 76                 & 327         & 90                 & 76                 & 8.6       & 90                 & 76                 & 0                    & 0                     & 4.1                 & 199                 \\ \hline
	P12-R19-6                  & 101                   & 197         & 0                  & 65                 & 339         & 35                 & 1                  & 7.3       & 36                 & 65                 & 1                    & 0                     & 15                  & 2.7                 \\ \hline
	P16-R12-8                  & 257                   & 626         & 1                  & 138                & 1285        & 0                  & 0                  & 54        & 116                & 141                & 118                  & 0                     & 3.3                 & -                   \\ \hline
	P15-R12-8                  & 223                   & 541         & 0                  & 117                & 1115        & 0                  & 0                  & 96        & 106                & 117                & 106                  & 0                     & 4.6                 & -                   \\ \hline
	P1-R18-15                  & 43                    & 52          & 0                  & 34                 & 0.5         & 0                  & 43                 & 0.3       & 0                  & 43                 & 0                    & 0                     & 50                  & 5.7                 \\ \hline
	P3-R18-15                  & 107                   & 493         & 0                  & 9                  & 12          & 77                 & 30                 & 5.8       & 78                 & 29                 & 0                    & 0                     & 175                 & 19                  \\ \hline
	\end{tabular}
\end{table*}

\begin{table*}[!htb]
	\caption{The results for Main Experiment (Hard Targets)}
\label{tab:hardtargets}
	\footnotesize
	\begin{tabular}{|l|r|r|r|r|r|r|r|r|r|r|r|r|r|r|}
	\hline
	\multirow{3}{*}{Benchmark} & \multirow{3}{*}{\#HT} & \multicolumn{3}{c|}{KLEE}                                            & \multicolumn{3}{c|}{CBMC}                                            & \multicolumn{3}{c|}{TracerX}                                         & \multirow{3}{*}{\#W} & \multirow{3}{*}{\#L} & Speedup & Speedup \\ \cline{3-11}

&                       & Time                       & \multirow{2}{*}{U} & \multirow{2}{*}{R} & Time                       & \multirow{2}{*}{U} & \multirow{2}{*}{R} & Time                       & \multirow{2}{*}{U} & \multirow{2}{*}{R} &                      &                       &         vs            &    vs                 \\ 
&                       & \multicolumn{1}{l|}{(min)} &                    &                    & \multicolumn{1}{l|}{(min)} &                    &                    & \multicolumn{1}{l|}{(min)} &                    &                    &                      &                       &          KLEE           &       CBMC              \\ \hline

	psyco1-100                 & 12                    & 120                        & 0                  & 0                  & 120                        & 0                  & 0                  & 0.7                        & 12                 & 0                  & 12                   & 0                     & -                   & -                   \\ \hline
psyco1-500                 & 12                    & 120                        & 0                  & 0                  & 120                        & 0                  & 0                  & 13                         & 12                 & 0                  & 12                   & 0                     & -                   & -                   \\ \hline
psyco1-1000                & 12                    & 120                        & 0                  & 0                  & 120                        & 0                  & 0                  & 60                         & 12                 & 0                  & 12                   & 0                     & -                   & -                   \\ \hline
psyco2-8                   & 20                    & 200                        & 0                  & 0                  & 36                         & 20                 & 0                  & 4.7                        & 20                 & 0                  & 0                    & 0                     & -                   & 7.6                 \\ \hline
psyco2-10                  & 20                    & 200                        & 0                  & 0                  & 68                         & 20                 & 0                  & 32                         & 20                 & 0                  & 0                    & 0                     & -                   & 2.1                 \\ \hline
psyco2-12                  & 20                    & 203                        & 0                  & 0                  & 115                        & 20                 & 0                  & 200                        & 5                  & 0                  & 0                    & \leavevmode\color{red}{15}                    & -                   & \leavevmode\color{red}{0.6}                 \\ \hline
psyco3-8                   & 20                    & 200                        & 0                  & 0                  & 36                         & 20                 & 0                  & 4.9                        & 20                 & 0                  & 0                    & 0                     & -                   & 7.3                 \\ \hline
psyco3-10                  & 20                    & 200                        & 0                  & 0                  & 68                         & 20                 & 0                  & 33                         & 20                 & 0                  & 0                    & 0                     & -                   & 2.1                 \\ \hline
psyco3-12                  & 20                    & 203                        & 0                  & 0                  & 117                        & 20                 & 0                  & 201                        & 0                  & 0                  & 0                    & \leavevmode\color{red}{20}                    & -                   & -                   \\ \hline
psyco4-8                   & 21                    & 210                        & 0                  & 0                  & 38                         & 21                 & 0                  & 4.9                        & 21                 & 0                  & 0                    & 0                     & -                   & 7.7                 \\ \hline
psyco4-10                  & 21                    & 210                        & 0                  & 0                  & 72                         & 21                 & 0                  & 34                         & 21                 & 0                  & 0                    & 0                     & -                   & 2.1                 \\ \hline
psyco4-12                  & 21                    & 212                        & 0                  & 0                  & 122                        & 21                 & 0                  & 207                        & 21                 & 0                  & 0                    & 0                     & -                   & \leavevmode\color{red}{0.6}                 \\ \hline
psyco5-100                 & 11                    & 110                        & 0                  & 0                  & 110                        & 0                  & 0                  & 0.7                        & 11                 & 0                  & 11                   & 0                     & -                   & -                   \\ \hline
psyco5-500                 & 11                    & 110                        & 0                  & 0                  & 110                        & 0                  & 0                  & 13                         & 11                 & 0                  & 11                   & 0                     & -                   & -                   \\ \hline
psyco5-1000                & 11                    & 110                        & 0                  & 0                  & 110                        & 0                  & 0                  & 60                         & 11                 & 0                  & 11                   & 0                     & -                   & -                   \\ \hline
psyco6-100                 & 3                     & 30                         & 0                  & 0                  & 26                         & 3                  & 0                  & 0.3                        & 3                  & 0                  & 0                    & 0                     & -                   & 86                  \\ \hline
psyco6-500                 & 3                     & 30                         & 0                  & 0                  & 30                         & 0                  & 0                  & 6.0                        & 3                  & 0                  & 3                    & 0                     & -                   & -                   \\ \hline
psyco6-1000                & 3                     & 30                         & 0                  & 0                  & 30                         & 0                  & 0                  & 29                         & 3                  & 0                  & 3                    & 0                     & -                   & -                   \\ \hline
psyco7-8                   & 44                    & 460                        & 0                  & 0                  & 440                        & 0                  & 0                  & 6.9                        & 44                 & 0                  & 44                   & 0                     & -                   & -                   \\ \hline
psyco7-10                  & 44                    & 464                        & 0                  & 0                  & 440                        & 0                  & 0                  & 9.1                        & 44                 & 0                  & 44                   & 0                     & -                   & -                   \\ \hline
psyco7-12                  & 44                    & 474                        & 0                  & 0                  & 440                        & 0                  & 0                  & 9.5                        & 44                 & 0                  & 44                   & 0                     & -                   & -                   \\ \hline
m34-C-A-6                  & 26                    & 260                        & 0                  & 0                  & 162                        & 14                 & 0                  & 1.5                        & 17                 & 9                  & 12                   & 0                     & -                   & -                   \\ \hline
m34-C-D-10                 & 2                     & 18                         & 0                  & 2                  & 20                         & 0                  & 0                  & 0.01                       & 0                  & 2                  & 0                    & 0                     & 1869                & -                   \\ \hline
m217-R-A-25                & 60                    & 505                        & 0                  & 27                 & 600                        & 0                  & 0                  & 419                        & 0                  & 26                 & 1                    & \leavevmode\color{red}{2}                     & 2.9                 & -                   \\ \hline
m217-R-D-3                & 88                    & 878                        & 0                  & 3                  & 880                        & 0                  & 0                  & 11                         & 84                 & 4                  & 85                   & 0                     & 81                  & -                   \\ \hline
P4-R12-8                   & 60                    & 600                        & 0                  & 0                  & 600                        & 0                  & 0                  & 54                         & 60                 & 0                  & 60                   & 0                     & -                   & -                   \\ \hline
P5-R12-8                   & 60                    & 601                        & 0                  & 0                  & 600                        & 0                  & 0                  & 267                        & 36                 & 24                 & 60                   & 0                     & -                   & -                   \\ \hline
P6-R12-8                   & 56                    & 557                        & 0                  & 1                  & 217                        & 35                 & 21                 & 178                        & 35                 & 21                 & 0                    & 0                     & 4                   & 1.5                 \\ \hline
P11-R19-50                 & 14                    & 140                        & 0                  & 0                  & 3.2                        & 14                 & 0                  & 140                        & 0                  & 0                  & 0                    & \leavevmode\color{red}{14}                    & -                   & -                   \\ \hline
P12-R19-10                 & 27                    & 272                        & 0                  & 0                  & 178                        & 11                 & 0                  & 87                         & 11                 & 16                 & 16                   & 0                     & -                   & \leavevmode\color{red}{0.2}                 \\ \hline
P17-R12-8                  & 36                    & 360                        & 0                  & 0                  & 360                        & 0                  & 0                  & 152                        & 36                 & 0                  & 36                   & 0                     & -                   & -                   \\ \hline
P2-T-R17-12                & 9                     & 90                         & 0                  & 0                  & 46                         & 7                  & 0                  & 0.9                        & 9                  & 0                  & 2                    & 0                     & -                   & 38                  \\ \hline
P3-R17-7                   & 73                    & 731                        & 0                  & 0                  & 713                        & 12                 & 0                  & 6.9                        & 73                 & 0                  & 61                   & 0                     & -                   & 90                  \\ \hline
P11-R17-7                  & 44                    & 440                        & 0                  & 0                  & 73                         & 40                 & 0                  & 2.8                        & 44                 & 0                  & 4                    & 0                     & -                   & 13                  \\ \hline
P14-R12-20                 & 35                    & 354                        & 0                  & 0                  & 329                        & 24                 & 0                  & 21                         & 35                 & 0                  & 11                   & 0                     & -                   & 15                  \\ \hline
m34-R-A-6                  & 96                    & 960                        & 0                  & 0                  & 550                        & 79                 & 0                  & 88                         & 96                 & 0                  & 17                   & 0                     & -                   & 5.2                 \\ \hline
m217-C-A-10                & 11                    & 110                        & 0                  & 0                  & 0.08                       & 0                  & 11                 & 0.05                       & 0                  & 11                 & 0                    & 0                     & -                   & 1.1                 \\ \hline
m217-C-D-10                & 24                    & 240                        & 0                  & 0                  & 0.6                        & 7                  & 17                 & 4.4                        & 7                  & 17                 & 0                    & 0                     & -                   & 3.1                 \\ \hline
P1-R19-20                  & 9                     & 90                         & 0                  & 0                  & 0.07                       & 0                  & 9                  & 17                         & 0                  & 8                  & 0                    & \leavevmode\color{red}{1}                     & -                   & \leavevmode\color{red}{0.9}                 \\ \hline
P2-R19-20                  & 6                     & 60                         & 0                  & 0                  & 40                         & 0                  & 6                  & 0.08                       & 0                  & 6                  & 0                    & 0                     & -                   & 575                 \\ \hline
P3-R19-15                  & 12                    & 120                        & 0                  & 0                  & 53                         & 12                 & 0                  & 121                        & 0                  & 0                  & 0                    & \leavevmode\color{red}{12}                    & -                   & -                   \\ \hline
P3-R12-8                   & 72                    & 722                        & 0                  & 0                  & 129                        & 72                 & 0                  & 6.1                        & 72                 & 0                  & 0                    & 0                     & -                   & 21                  \\ \hline
P12-R19-6                  & 15                    & 150                        & 0                  & 0                  & 6.5                        & 15                 & 0                  & 2.4                        & 15                 & 0                  & 0                    & 0                     & -                   & 2.7                 \\ \hline
P16-R12-8                  & 119                   & 1181                       & 0                  & 1                  & 1190                       & 0                  & 0                  & 40                         & 116                & 3                  & 118                  & 0                     & 3.7                 & -                   \\ \hline
P15-R12-8                  & 111                   & 1061                       & 0                  & 5                  & 1110                       & 0                  & 0                  & 84                         & 106                & 5                  & 106                  & 0                     & 2.2                 & -                   \\ \hline
P1-R18-15                  & 3                     & 23                         & 0                  & 2                  & 0.03                       & 0                  & 3                  & 0.02                       & 0                  & 3                  & 0                    & 0                     & 1470                & 3.1                 \\ \hline
P3-R18-15                  & 9                     & 90                         & 0                  & 0                  & 1.1                        & 5                  & 4                  & 0.4                        & 5                  & 4                  & 0                    & 0                     & -                   & 13                  \\ \hline
	\end{tabular}
\end{table*}

\normalsize
\begin{figure}[!]
	\centering
	\begin{subfigure}{.5\textwidth}
		\centering
		\includegraphics[width=1\linewidth]{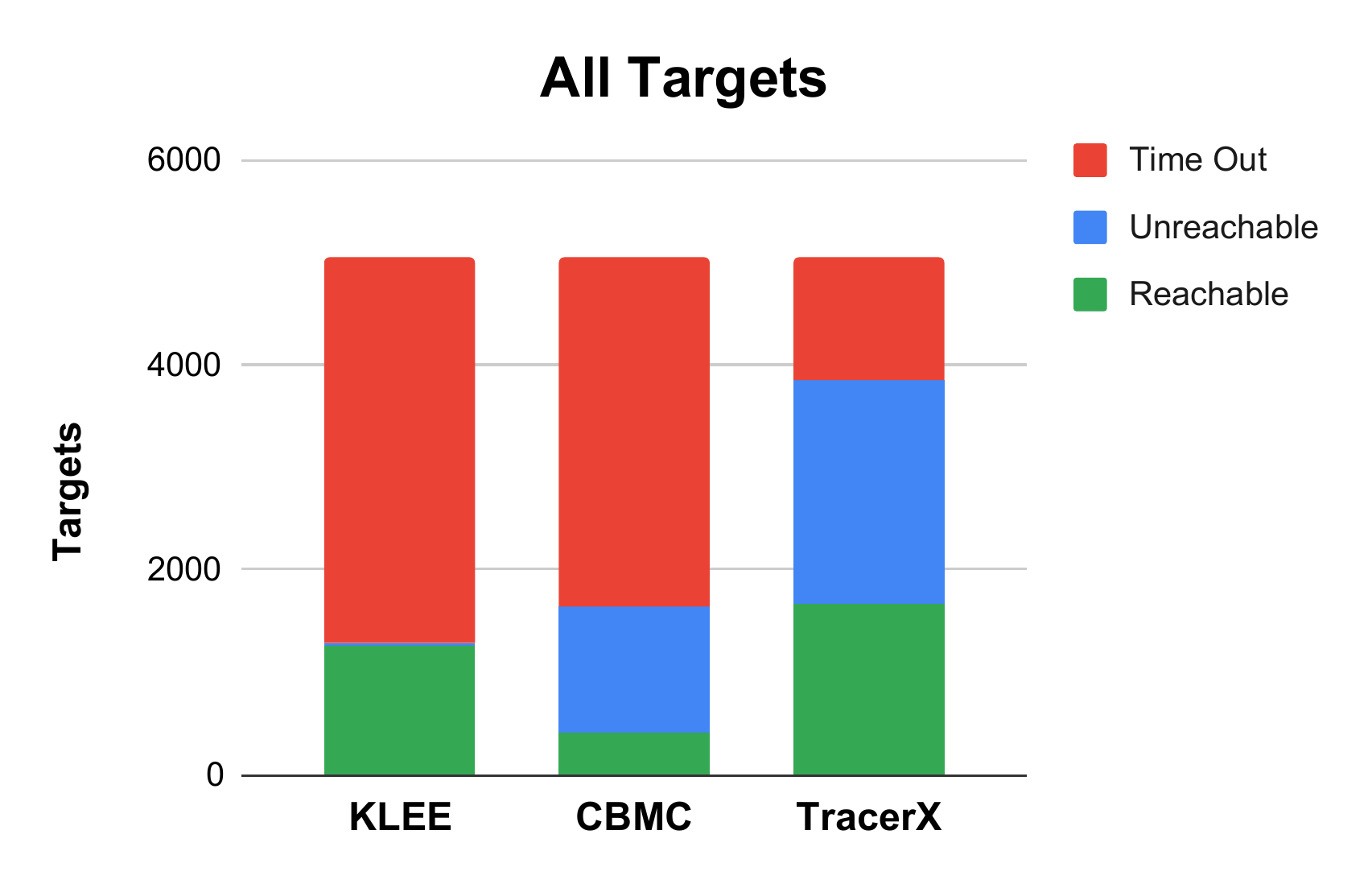}
		\caption{Aggregated numbers of \textit{TO}, \textit{U}, and \textit{R}}
		\label{fig:AllTargetscount}		
	\end{subfigure}%
	\begin{subfigure}{.5\textwidth}
		\centering
		\includegraphics[width=1\linewidth]{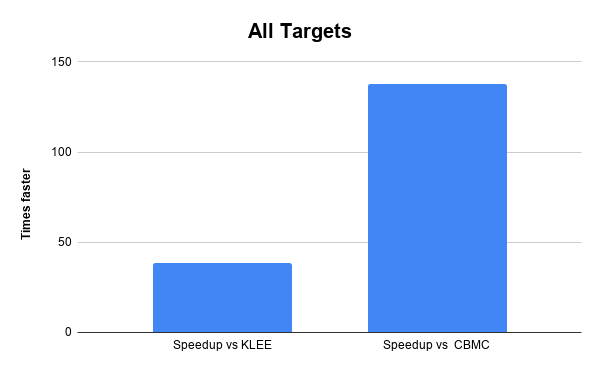}		
		\caption{Aggregated Speed}
		\label{fig:AllTargetsSpeed}
	\end{subfigure}
	\caption{Aggregated results for All target experiment}
	\label{fig:AllTargets}
	\vspace{-4mm}
\end{figure}

\begin{figure}
	\centering
	\begin{subfigure}{.5\textwidth}
		\centering
		\includegraphics[width=1\linewidth]{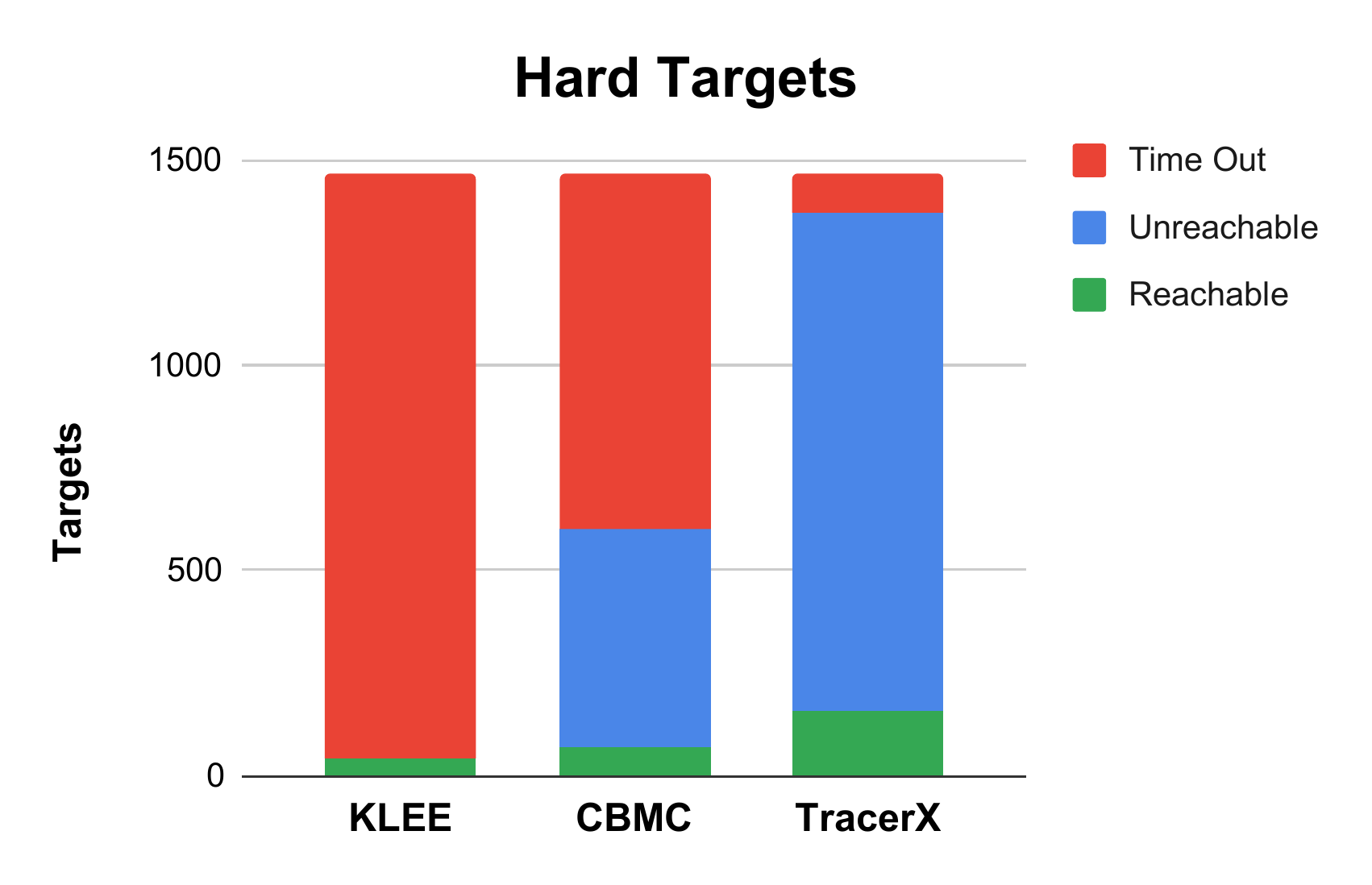}
		\caption{Aggregated numbers of \textit{TO}, \textit{U}, and \textit{R}}
		\label{fig:HardTargetscount}
	\end{subfigure}%
	\begin{subfigure}{.5\textwidth}
		\centering
		\includegraphics[width=1\linewidth]{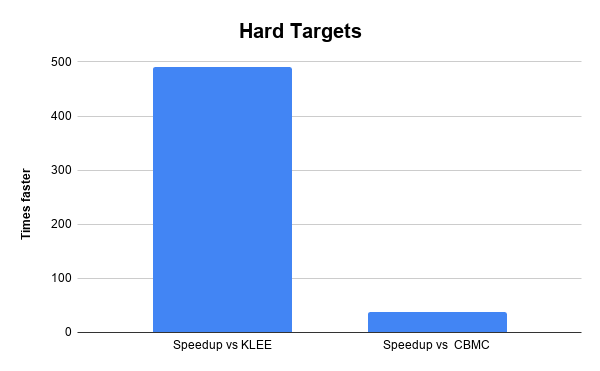}		
		\caption{Aggregated Speed}
		\label{fig:HardTargetsSpeed}
	\end{subfigure}
	\caption{Aggregated results for Hard target experiment}
	\label{fig:HardTargets}
	\vspace{-4mm}
\end{figure}

%\vspace{-1mm}
\noindent
\subsection{Main Experiment (Penetration)}
The main purpose of this experiment is to detect individually each target (bug) injected in 
the program. Some of these targets will be easy to reach and some are very difficult. Also,
some of the targets are located in unreachable parts of the program. 
We compare \tracerx{} and the baseline approaches on their capability in detecting easy as well as hard targets. 

We determine a subset of all targets as hard targets via a filtering phase. These are obtained 
by filtering out targets which can be proved easily by state-of-the-art methods. We filtered out all 
the targets that are detected by KLEE within 5 minutes. Moreover, in a second step, we filtered
out any target which can be determined as unreachable by {\tt Framma-C} \cite{Framma-C} (a sound 
static analyzer). We consider the remaining targets as hard targets.

See Table \ref{tab:alltargets} which presents the results of our experiment on all
targets. The KLEE column reports KLEE running in its random mode  \cite{cadar08klee} 
(the choice mode for code coverage), and the CBMC column reports CBMC running in its default 
mode. \tracerx{} column reports \tracerx{} running using depth-first search (DFS). 
We annotated each target separately in the benchmarks  and ran the tools on 
each of the annotated programs for 5 minutes.

The \textit{Benchmark} column reports the benchmark program names.
The \textit{\#AT} column in Table \ref{tab:alltargets} reports the total number of targets\footnote{A 
target means assert(0), which is a runtime error.}. 
The \textit{Time} column for each tool reports the aggregated execution
time (in minutes) for all the targets injected in a program. The columns \textit{U} and \textit{R} 
report the number of targets which have been proven to be unreachable or reachable. 
The tools hit timeout on the remaining targets (not able to prove them as unreachable or 
reachable).   The remaining targets, those which caused a timeout,
can be computed by $\#AT - (U + R)$ in Table \ref{tab:alltargets}.
%and by $\#HT - (U + R)$ in Table  \ref{tab:hardtargets}.

The WIN column (indicated by \#W) reports the numbers of targets that \tracerx{} was able to prove while
none of the baseline tools were able to prove. In contrast, the LOSE column (indicated by \#L) reports the numbers
of targets that either of the baseline tools were able to prove while \tracerx{} was not able to prove it.
The last two columns (Speedup vs KLEE) and (Speedup vs CBMC)
 show the relative speed of \tracerx{} over KLEE and CBMC 
respectively, over targets that have been proved reachable or unreachable
by KLEE/CBMC and \tracerx{}. When a row is marked with ``-'' it means one of the tools had hit timeout on all
targets in that program. 
The value denotes the relative speedup: for example, $0.5$ means \tracerx{} was half as fast,
and $2.0$ means \tracerx{} was twice as fast\footnote{The lose cases for \tracerx{} have been highlighted with red color in Tables \ref{tab:alltargets} and \ref{tab:hardtargets}.}.

Fig. \ref{fig:AllTargets} shows the aggregated results of the all targets experiment. 
Fig. \ref{fig:AllTargetscount}  shows the total numbers of targets that each tool has been able to prove as 
reachable or unreachable. The remaining targets are the ones where the tools timeout.
Moreover, in Fig. \ref{fig:AllTargetsSpeed},  we present the aggregate on the relative speedup of \tracerx{} 
over KLEE and CBMC.
Finally, note that the detailed information on each target can be seen in \cite{All-Targets}.  

See Tables \ref{tab:hardtargets} which present the results of our experiment on hard
targets. The columns are the same as in Table \ref{tab:alltargets} except for  the  \textit{\#HT} column
which  reports the total number of hard targets. 
Since, in the previous experiments, the timeout was set to 5 minutes,  for the tools to have 
a higher chance of finding hard targets we have extended the timeout to 10 minutes.

Fig. \ref{fig:HardTargets} shows the aggregated results of the Hard targets experiment. 
Fig. \ref{fig:HardTargetscount}  shows
the total number of targets that each tool has been able to prove as reachable or
unreachable. The remaining targets are the ones where the tools timeout.
Moreover, in Fig. \ref{fig:HardTargetsSpeed},  we present the aggregate on the 
relative speedup of \tracerx{}  over KLEE and CBMC. Finally, we reported the  detailed
information on each hard target ran for all the programs in \cite{Hard-Targets}.

%% file: supl_experiment.tex
\ignore{
\textcolor{teal}{In this section we present our Supplementary Experiment (Exploration). 
First part will the exploration of Bug finding and second will be the
exploration of Block coverage. In first part we will consider All
Targets and Hard Targets program to report the first target proved as
reachable. Whereas, in second part we show the block coverage of same
set of SV-COMP programs. Also, we reported the block coverage results
on GNU Coreutils program too.}
}

\begin{table*}[!htb]
	\caption{The results for Supplementary Bug Finding Experiment (All Targets and Hard Targets)}
	\label{tab:supplm1m2}
	\vspace*{-3mm}
	\scriptsize
	\begin{tabular}{|l|r|c|r|c|r|c|r|c|r|c|r|c|r|c|r|c|}
		%\scriptsize
		%\begin{longtable}[c]{|c|c|c|c|c|c|c|c|c|c|c|c|c|c|c|c|c|}
		\hline
		\multirow{4}{*}{Benchmark} & \multicolumn{8}{c|}{All Targets}                                                                                                                         & \multicolumn{8}{c|}{Hard Targets}                                                                                                     \\ \cline{2-17} 
		& \multicolumn{2}{c|}{KLEE}       & \multicolumn{2}{c|}{CBMC}                            & \multicolumn{2}{c|}{TracerX-D} & \multicolumn{2}{c|}{TracerX-R} & \multicolumn{2}{c|}{KLEE}       & \multicolumn{2}{c|}{CBMC}       & \multicolumn{2}{c|}{TracerX-D}  & \multicolumn{2}{c|}{TracerX-R}  \\ \cline{2-17} 
		& Time     & \multirow{2}{*}{1/0} & Time     & \multicolumn{1}{l|}{\multirow{2}{*}{1/0}} & Time   & \multirow{2}{*}{1/0}  & Time   & \multirow{2}{*}{1/0}  & Time     & \multirow{2}{*}{1/0} & Time     & \multirow{2}{*}{1/0} & Time     & \multirow{2}{*}{1/0} & Time     & \multirow{2}{*}{1/0} \\
		
		& (sec)    &                      & (sec)    & \multicolumn{1}{l|}{}                     & (sec)  &                       & (sec)  &                       & (sec)    &                      & (sec)    &                      & (sec)    &                      & (sec)    &                      \\ \hline
		psyco1-100                 & 0.01     & 1                    & $\infty$ & 0                                         & 0.22   & 1                     & 0.02   & 1                     & $\infty$ & 0                    & $\infty$ & 0                    & 5.6      & 0                    & $\infty$ & 0                    \\ \hline
		psyco1-500                 & 0.02     & 1                    & $\infty$ & 0                                         & 2.77   & 1                     & 0.02   & 1                     & $\infty$ & 0                    & $\infty$ & 0                    & 100      & 0                    & $\infty$ & 0                    \\ \hline
		psyco1-1000                & 0.02     & 1                    & $\infty$ & 0                                         & 10.4   & 1                     & 0.02   & 1                     & $\infty$ & 0                    & $\infty$ & 0                    & 451      & 0                    & $\infty$ & 0                    \\ \hline
		psyco2-8                   & 0.05     & 1                    & 155      & 1                                         & 0.04   & 1                     & 0.05   & 1                     & 878      & 0                    & 152      & 0                    & 20       & 0                    & 70       & 0                    \\ \hline
		psyco2-10                  & 0.03     & 1                    & 295      & 1                                         & 0.03   & 1                     & 0.03   & 1                     & $\infty$ & 0                    & 293      & 0                    & 132      & 0                    & 870      & 0                    \\ \hline
		psyco2-12                  & 0.03     & 1                    & 535      & 1                                         & 0.04   & 1                     & 0.04   & 1                     & $\infty$ & 0                    & 504      & 0                    & 844      & 0                    & $\infty$ & 0                    \\ \hline
		psyco3-8                   & 0.03     & 1                    & 163      & 1                                         & 0.04   & 1                     & 0.03   & 1                     & 879      & 0                    & 151      & 0                    & 20       & 0                    & 70       & 0                    \\ \hline
		psyco3-10                  & 0.03     & 1                    & 296      & 1                                         & 0.04   & 1                     & 0.04   & 1                     & $\infty$ & 0                    & 293      & 0                    & 133      & 0                    & 888      & 0                    \\ \hline
		psyco3-12                  & 0.03     & 1                    & 520      & 1                                         & 0.04   & 1                     & 0.03   & 1                     & $\infty$ & 0                    & 506      & 0                    & 853      & 0                    & $\infty$ & 0                    \\ \hline
		psyco4-8                   & 0.04     & 1                    & 154      & 1                                         & 0.04   & 1                     & 0.03   & 1                     & 875      & 0                    & 152      & 0                    & 20       & 0                    & 69       & 0                    \\ \hline
		psyco4-10                  & 0.03     & 1                    & 306      & 1                                         & 0.3    & 1                     & 0.03   & 1                     & $\infty$ & 0                    & 291      & 0                    & 2696     & 0                    & $\infty$ & 0                    \\ \hline
		psyco4-12                  & 0.03     & 1                    & 518      & 1                                         & 0.04   & 1                     & 0.03   & 1                     & $\infty$ & 0                    & 506      & 0                    & 840      & 0                    & $\infty$ & 0                    \\ \hline
		psyco5-100                 & 0.02     & 1                    & 1893     & 1                                         & 0.22   & 1                     & 0.02   & 1                     & $\infty$ & 0                    & 1851     & 0                    & 5.8      & 0                    & $\infty$ & 0                    \\ \hline
		psyco5-500                 & 0.03     & 1                    & $\infty$ & 0                                         & 2.71   & 1                     & 0.02   & 1                     & $\infty$ & 0                    & $\infty$ & 0                    & 113      & 0                    & $\infty$ & 0                    \\ \hline
		psyco5-1000                & 0.02     & 1                    & $\infty$ & 0                                         & 10.0   & 1                     & 0.02   & 1                     & $\infty$ & 0                    & $\infty$ & 0                    & 478      & 0                    & $\infty$ & 0                    \\ \hline
		psyco6-100                 & 0.04     & 1                    & 772      & 1                                         & 0.25   & 1                     & 0.04   & 1                     & $\infty$ & 0                    & 692      & 0                    & 9        & 0                    & $\infty$ & 0                    \\ \hline
		psyco6-500                 & 0.04     & 1                    & $\infty$ & 0                                         & 2.79   & 1                     & 0.04   & 1                     & $\infty$ & 0                    & $\infty$ & 0                    & 188      & 0                    & $\infty$ & 0                    \\ \hline
		psyco6-1000                & 0.04     & 1                    & $\infty$ & 0                                         & 10.4   & 1                     & 0.03   & 1                     & $\infty$ & 0                    & $\infty$ & 0                    & 863      & 0                    & $\infty$ & 0                    \\ \hline
		psyco7-8                   & $\infty$ & 0                    & 1278     & 0                                         & 11.4   & 0                     & 1251   & 0                     & $\infty$ & 0                    & 1280     & 0                    & 11       & 0                    & 1065     & 0                    \\ \hline
		psyco7-10                  & $\infty$ & 0                    & 2442     & 0                                         & 15.7   & 0                     & 1040   & 0                     & $\infty$ & 0                    & 2417     & 0                    & 15       & 0                    & 1185     & 0                    \\ \hline
		psyco7-12                  & $\infty$ & 0                    & $\infty$ & 0                                         & 19.0   & 0                     & 1820   & 0                     & $\infty$ & 0                    & $\infty$ & 0                    & 19       & 0                    & 1203     & 0                    \\ \hline
		m34-C-A-6                  & 0.22     & 1                    & $\infty$ & 0                                         & 0.03   & 1                     & 0.09   & 1                     & 1973     & 1                    & $\infty$ & 0                    & 14       & 1                    & 354      & 1                    \\ \hline
		m34-C-D-10                 & 0.31     & 1                    & $\infty$ & 0                                         & 0.1    & 1                     & 0.23   & 1                     & 921      & 1                    & $\infty$ & 0                    & 0.5      & 1                    & 155      & 1                    \\ \hline
		m217-R-A-25                & 0.07     & 1                    & $\infty$ & 0                                         & 0.14   & 1                     & 0.06   & 1                     & 410      & 1                    & $\infty$ & 0                    & 850      & 1                    & $\infty$ & 0                    \\ \hline
		m217-R-D-3                 & 0.98     & 1                    & $\infty$ & 0                                         & 0.11   & 1                     & 0.54   & 1                     & 38       & 1                    & $\infty$ & 0                    & $\infty$ & 0                    & 189      & 1                    \\ \hline
		P4-R12-8                   & 0.27     & 1                    & 1684     & 1                                         & 0.43   & 1                     & 0.39   & 1                     & $\infty$ & 0                    & 1689     & 0                    & 68       & 0                    & 463      & 0                    \\ \hline
		P5-R12-8                   & 0.66     & 1                    & 508      & 1                                         & 2.05   & 1                     & 3.18   & 1                     & $\infty$ & 0                    & 545      & 1                    & 271      & 1                    & 1755     & 1                    \\ \hline
		P6-R12-8                   & 0.5      & 1                    & 285      & 1                                         & 0.92   & 1                     & 0.78   & 1                     & 496      & 1                    & 278      & 1                    & 97       & 1                    & 237      & 1                    \\ \hline
		P11-R19-50                 & 0.02     & 1                    & 13       & 1                                         & 0.23   & 1                     & 0.02   & 1                     & $\infty$ & 0                    & 95       & 0                    & $\infty$ & 0                    & $\infty$ & 0                    \\ \hline
		P12-R19-10                 & 0.06     & 1                    & $\infty$ & 0                                         & 0.04   & 1                     & 0.06   & 1                     & 1018  & 1                    & $\infty$ & 0                    & 4.9      & 1                    & 171      & 1                    \\ \hline
		P17-R12-8                  & 0.05     & 1                    & $\infty$ & 0                                         & 0.24   & 1                     & 0.05   & 1                     & $\infty$ & 0                    & $\infty$ & 0                    & 821      & 0                    & $\infty$ & 0                    \\ \hline
		P2-T-R17-12                & 0.03     & 1                    & $\infty$ & 0                                         & 0.03   & 1                     & 0.02   & 1                     & $\infty$ & 0                    & 170      & 0                    & $\infty$ & 0                    & $\infty$ & 0                    \\ \hline
		P3-R17-7                   & 0.08     & 1                    & $\infty$ & 0                                         & 0.03   & 1                     & 0.08   & 1                     & $\infty$ & 0                    & 1559     & 0                    & $\infty$ & 0                    & $\infty$ & 0                    \\ \hline
		P11-R17-7                  & 0.03     & 1                    & $\infty$ & 0                                         & 0.01   & 1                     & 0.03   & 1                     & $\infty$ & 0                    & $\infty$ & 0                    & $\infty$ & 0                    & $\infty$ & 0                    \\ \hline
		P14-R12-20                 & 0.15     & 1                    & 2026     & 1                                         & 0.33   & 1                     & 0.02   & 1                     & $\infty$ & 0                    & 1877     & 0                    & 167      & 0                    & $\infty$ & 0                    \\ \hline
		m34-R-A-6                  & 0.22     & 1                    & $\infty$ & 0                                         & 0.05   & 1                     & 0.18   & 1                     & $\infty$ & 0                    & $\infty$ & 0                    & $\infty$ & 0                    & $\infty$ & 0                    \\ \hline
		m217-C-A-10                & 0.22     & 1                    & 6        & 1                                         & 0.02   & 1                     & 0.1    & 1                     & 849      & 1                    & 4        & 1                    & 0.3      & 1                    & 1279     & 1                    \\ \hline
		m217-C-D-10                & 0.12     & 1                    & 3        & 1                                         & 0.03   & 1                     & 0.04   & 1                     & 1059     & 1                    & 2        & 1                    & 0.3      & 1                    & 1226     & 1                    \\ \hline
		P1-R19-20                  & 0.02     & 1                    & 1        & 1                                         & \leavevmode\color{blue}{0}      & 1                     & 0.01   & 1                     & $\infty$ & 0                    & 0.8      & 1                    & 0.3      & 1                    & 3131     & 1                    \\ \hline
		P2-R19-20                  & 0.06     & 1                    & 719      & 1                                         & 0.04   & 1                     & 0.07   & 1                     & $\infty$ & 0                    & 678      & 1                    & 11       & 1                    & 464      & 1                    \\ \hline
		P3-R19-15                  & 0.1      & 1                    & $\infty$ & 0                                         & 0.05   & 1                     & 0.19   & 1                     & $\infty$ & 0                    & 1463     & 0                    & $\infty$ & 0                    & $\infty$ & 0                    \\ \hline
		P3-R12-8                   & 0.02     & 1                    & 178      & 1                                         & 0.16   & 1                     & 0.02   & 1                     & $\infty$ & 0                    & 169      & 0                    & 15       & 0                    & 124      & 0                    \\ \hline
		P12-R19-6                  & 0.05     & 1                    & $\infty$ & 0                                         & 0.03   & 1                     & 0.06   & 1                     & 2881     & 0                    & 86       & 0                    & 1573     & 0                    & 2137     & 0                    \\ \hline
		P16-R12-8                  & 0.02     & 1                    & 2988     & 1                                         & 0.13   & 1                     & 0.02   & 1                     & 31       & 1                    & 2940     & 1                    & 6        & 1                    & 40       & 1                    \\ \hline
		P15-R12-8                  & 0.03     & 1                    & $\infty$ & 0                                         & 0.13   & 1                     & 0.03   & 1                     & 1.58     & 1                    & $\infty$ & 0                    & 10       & 1                    & 14       & 1                    \\ \hline
		P1-R18-15                  & 0.01     & 1                    & 1        & 1                                         & \leavevmode\color{blue}{0}      & 1                     & 0.01   & 1                     & 405      & 1                    & 1        & 1                    & 96       & 1                    & $\infty$ & 0                    \\ \hline
		P3-R18-15                  & 0.02     & 1                    & 7        & 1                                         & \leavevmode\color{blue}{0}      & 1                     & 0.01   & 1                     & $\infty$ & 0                    & 7        & 1                    & 2        & 1                    & $\infty$ & 0                    \\ \hline
		%\end{longtable}
	\end{tabular}
	\vspace*{-6mm}
\end{table*}

\subsubsection{Bug Finding}

\ \\
\ignore{
\textcolor{teal}{
Table \ref{tab:supplm1m2} shows the results of bug finding. This
experiment has the theme of standard \textit{TestComp} competition
event. Basically, the C program has multiple injected targets and
tools are required to report the very first target and then terminate
the process. Suppose that the targets are \textit{low hanging fruits}
i.e. easy, so any random strategy based tool can find them
easily. But, if these targets are the hard one then it is a
challenging task for any tool. Also, if the target is unreachable then
a tool required a full exercise of SET. The Table \ref{tab:supplm1m2}
has three main columns \textit{Benchmark}, \textit{All Targets},
and \textit{Hard Targets}. Now, columns \textit{All Targets},
and \textit{Hard Targets} split into four sub-columns viz.  KLEE,
CBMC, \tracerx{}-D, and \tracerx{}-R. These sub-columns are further
split into \textit{Time} and \textit{1/0}.}  }

Table \ref{tab:supplm1m2} shows the results of bug finding.
{\tracerx{}-D is our system running under a Depth First Search (DFS)
strategy and \tracerx{}-R with a random (KLEE-like) strategy.  The
column \textit{Time} is in seconds, \textit{1/0} shows whether the
target was proven (reachable or unreachable).  Here, \textquotedblleft
1" means that the target was proved reachable, and \textquotedblleft 0"
shows that the target is unreachable \emph{if there was no timeout}.
% Now, if a target
%is \textquotedblleft 0" and the program is not timed out ($\infty$)
%then the target is proved as unreachable.  
For this experiment, the timeout was set to 1 hour.

Fig. \ref{fig:Bugfinding} shows aggregated results for Table \ref{tab:supplm1m2}.
%Figures \ref{fig:BugfindingAllTargets} and \ref{fig:BugfindingHardTargets} show the aggregated results for
%All Targets and hard Targets respectively. 
The height of the bar denotes the number of ``wins'' for each system.
A system wins when it proves a target faster than the others.

\begin{table*}[!htb]
	\caption{The results for Supplementary LLVM Block Coverage	Experiment} 
	\label{tab:supplblock-1}
	\vspace*{-3mm}
	\scriptsize
	\begin{tabular}{|l|r|r|r|r|r|r||l|r|r|r|r|r|r|}
%\scriptsize
%\begin{longtable}[c]{|l|r|r|r|r|r|r||l|r|r|r|r|r|r|}
 \hline \multirow{3}{*}{Benchmark}
	& \multicolumn{2}{c|}{KLEE} & \multicolumn{2}{c|}{TracerX-D}
	& \multicolumn{2}{c||}{TracerX-R} & \multirow{3}{*}{Benchmark}
	& \multicolumn{2}{c|}{KLEE} & \multicolumn{2}{c|}{TracerX-D}
	& \multicolumn{2}{c|}{TracerX-R} \\ \cline{2-7} \cline{9-14} &
	Time & \multirow{2}{*}{BB} & Time & \multirow{2}{*}{BB} & Time
	& \multirow{2}{*}{BB} & & Time & \multirow{2}{*}{BB} & Time
	& \multirow{2}{*}{BB} & Time & \multirow{2}{*}{BB} \\
	
	& (sec) &                     & (sec)   &                      & (sec)   &                      &
	& (sec) &                     & (sec)   &                      & (sec)   &                      \\ \hline

	psyco2-8                   & 845   & 24.53               & 15      & 24.53                & 55      & 24.53                &%\\ \hline
	psyco3-8                   & 842   & 24.62               & 15      & 24.62                & 57      & 24.62                \\ \hline
	psyco4-8                   & 845   & 24.53               & 15      & 24.53                & 57      & 24.53                &%\\ \hline
	P12-R19-6                  & 2686  & 80.99               & 1229    & 80.99                & 1978    & 80.99                \\ \hline
	m217-R-D-3                 & 2615  & 2.47                & $\infty$       & 1.86                 & $\infty$       & 2.47                 &%\\ \hline
	P3-R12-8                   & $\infty$     & 77.23               & 11      & 77.23                & 122     & 77.23                \\ \hline
	P4-R12-8                   & $\infty$     & 86.57               & 55      & 90.10                & 468     & 90.10                &%\\ \hline
	P5-R12-8                   & $\infty$     & 77.89               & 295     & 84.03                & 2494    & 84.03                \\ \hline
	P6-R12-8                   & $\infty$     & 79.18               & 254     & 82.46                & 2043    & 82.46               &%\\ \hline
	P16-R12-8                  & $\infty$     & 76.65               & 54      & 77.11                & 352     & 77.11                \\ \hline
	psyco7-8                   & $\infty$     & 13.79               & 9       & 16.39                & 1183    & 16.39                &%\\ \hline
	P15-R12-8                  & $\infty$     & 84.27               & 128     & 84.27                & 1007    & 84.27                \\ \hline
	psyco2-10                  & $\infty$     & 24.53               & 101     & 24.53                & 672     & 24.53                &%\\ \hline
	psyco3-10                  & $\infty$     & 24.62               & 102     & 24.62                & 706     & 24.62                \\ \hline
	psyco1-100                 & $\infty$     & 56.95               & 4.0     & 56.95                & $\infty$       & 56.95                &%\\ \hline
	psyco1-500                 & $\infty$     & 56.95               & 75      & 56.95                & $\infty$       & 56.95                \\ \hline
	psyco1-1000                & $\infty$     & 56.95               & 373     & 56.95                & $\infty$       & 56.95                &%\\ \hline
	psyco2-B12                 & $\infty$     & 24.53               & 651     & 24.53                & $\infty$       & 24.53                \\ \hline
	psyco3-B12                 & $\infty$     & 24.62               & 647     & 24.62                & $\infty$       & 24.62                &%\\ \hline
	psyco4-10                  & $\infty$     & 24.53               & 101     & 24.53                & $\infty$       & 24.53                \\ \hline
	psyco4-12                  & $\infty$     & 24.53               & 642     & 24.53                & $\infty$       & 24.53                &%\\ \hline
	psyco5-100                 & $\infty$     & 54.19               & 4.5     & 54.19                & $\infty$       & 54.19                \\ \hline
	psyco5-500                 & $\infty$     & 54.19               & 82      & 54.19                & $\infty$       & 54.19                &%\\ \hline
	psyco5-1000                & $\infty$     & 54.19               & 400     & 54.19                & $\infty$       & 54.19                \\ \hline
	psyco6-100                 & $\infty$     & 59.06               & 6.9     & 59.06                & $\infty$       & 59.06                &%\\ \hline
	psyco6-500                 & $\infty$     & 59.06               & 136     & 59.06                & $\infty$       & 59.06                \\ \hline
	psyco6-1000                & $\infty$     & 59.06               & 722     & 59.06                & $\infty$       & 59.06                &%\\ \hline
	P17-R12-8                  & $\infty$     & 83.50               & 653     & 83.59                & $\infty$       & 83.59                \\ \hline
	m217-C-D-10                & $\infty$     & 16.69               & 441     & 40.29                & $\infty$       & 22.45                &%\\ \hline
	psyco7-10                  & $\infty$     & 13.62               & 12      & 16.39                & $\infty$       & 16.39                \\ \hline
	psyco7-12                  & $\infty$     & 13.50               & 14      & 16.39                & $\infty$       & 16.39               &%\\ \hline
	P14-R12-20                 & $\infty$     & 96.55               & 128     & 96.55                & $\infty$       & 96.55                \\ \hline
	P1-R19-20                  & $\infty$     & 58.43               & $\infty$       & 64.79                & $\infty$       & 52.43               &%\\ \hline
	P2-R19-20                  & $\infty$     & 82.36               & $\infty$       & 87.58                & $\infty$       & 77.76                \\ \hline
	P3-R19-15                  & $\infty$     & 62.37               & $\infty$       & 62.37                & $\infty$       & 62.37                &%\\ \hline
	m217-R-A-25                & $\infty$     & 69.93               & $\infty$       & 51.80                & $\infty$       & 32.83                \\ \hline
	m34-C-A-6                  & $\infty$     & 18.53               & $\infty$       & 25.67                & $\infty$       & 18.53                &%\\ \hline
	m34-C-D-10                 & $\infty$     & 2.63                & $\infty$       & 8.36                 & $\infty$       & 3.78                 \\ \hline
	m34-R-A-6                  & $\infty$     & 31.36               & $\infty$       & 11.43                & $\infty$       & 22.67                &%\\ \hline
	m217-C-A-10                & $\infty$     & 50.51               & $\infty$       & 61.62                & $\infty$       & 50.51                \\ \hline
	P11-R19-50                 & $\infty$     & 82.88               & $\infty$       & 79.67                & $\infty$       & 82.88                &%\\ \hline
	P12-R19-10                 & $\infty$     & 80.99               & $\infty$       & 80.99                & $\infty$       & 82.58                \\ \hline
	P2-T-R17-12                & $\infty$     & 65.82               & $\infty$       & 65.82                & $\infty$       & 65.82                &%\\ \hline
	P3-R17-7                   & $\infty$     & 57.70               & $\infty$       & 55.95                & $\infty$       & 57.70                \\ \hline
	P11-R17-7                  & $\infty$     & 84.78               & $\infty$       & 84.78                & $\infty$       & 84.78                &%\\ \hline
	P1-R18-15                  & $\infty$     & 94.33               & $\infty$       & 51.27                & $\infty$       & 90.65                \\ \hline
	P3-R18-15                  & $\infty$     & 18.16               & $\infty$       & 23.95                & $\infty$       & 17.50                & - & - & - & - & - & - & - \\ \hline
%\end{longtable}
	\end{tabular}
\end{table*}

\normalsize
\begin{figure}[!htb]
	\centering
	\begin{subfigure}{.5\textwidth}
		\centering
		\includegraphics[width=1\linewidth]{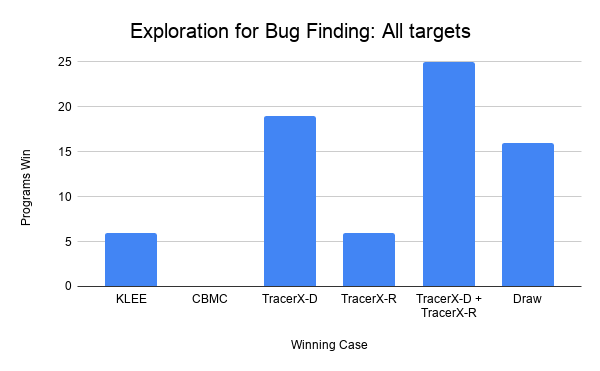}
		%\caption{Bug Finding: All Targets}
		\label{fig:BugfindingAllTargets}
	\end{subfigure}%
	\begin{subfigure}{.5\textwidth}
		\centering
		\includegraphics[width=1\linewidth]{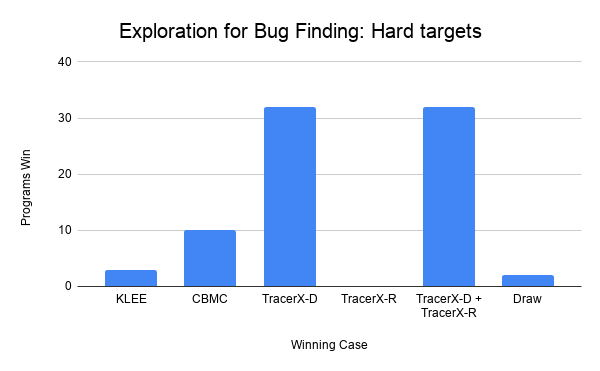}
		%\caption{Bug Finding: Hard Targets}
		\label{fig:BugfindingHardTargets}
	\end{subfigure}
	\vspace{-8mm}
	\caption{Aggregated results for Bug Finding}
	\label{fig:Bugfinding}
	\vspace{-2mm}
\end{figure}

\normalsize
\subsubsection{Block Coverage} 

\ \\

We consider LLVM Basic Block Coverage (\textit{BB}) as our coverage metric.  
Therefore we shall only compare against KLEE is this sub-experiment.

\ignore{
If a system finishes within timeout so, the reported block coverage is the maximum possible.
Now, if all system finish within timeout then the total execution time
consumed by the system should be compared, and if one of the baseline
systems has less execution time consumed so the system considered to
be the winner. Also, when none of the systems finish the execution
then the higher \textit{BB} by either of the systems will be the
winner.
}

Table \ref{tab:supplblock-1} shows the results of coverage achieved on 47 SV-COMP programs.  
The column \textit{BB} shows the block coverage percentage.
Columns 1 and 5 show the \textit{Benchmark} names. Columns
2 to 4 \& 6 to 8 show KLEE, \tracerx{}-D and \tracerx{}-R. 
The columns \textit{Time} in seconds.
Timeout set at 1 hour ($\infty$ in the table). 

We have also experimented with the Coreutils benchmark, for which KLEE
is famous for proving good coverage.
For space reasons, we relegate the detailed results to the appendix,
in Table \ref{tab:coreutils}.
Instead, Fig. \ref{fig:Coreutils} gives an overall picture of comparison
with KLEE, and also of comparison between using a DFS or random
strategy.  
% It is with this benchmark that the random strategy is quite helpful

Finally, see the aggregate results for coverage on both the SV-COMP
and Coreutils benchmarks in Fig. \ref{fig:explorationcoverage}.

\normalsize
\begin{figure}[!htb]
	\centering
	\begin{subfigure}{.5\textwidth}
		\centering
		\includegraphics[width=1\linewidth]{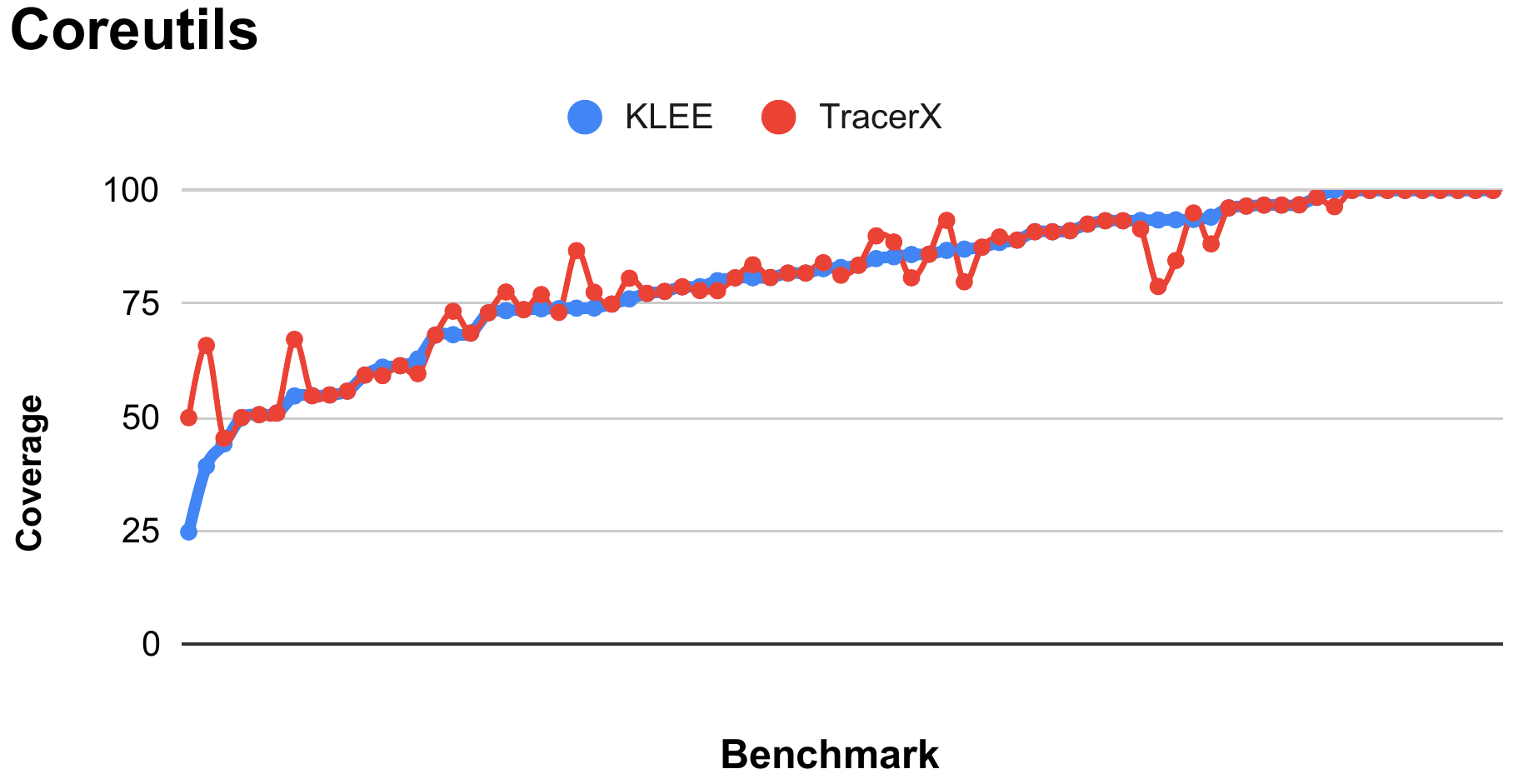}
		\caption{KLEE vs. TracerX}
		\label{fig:NonCoreutils}
	\end{subfigure}%
	\begin{subfigure}{.5\textwidth}
		\centering
		\includegraphics[width=1\linewidth]{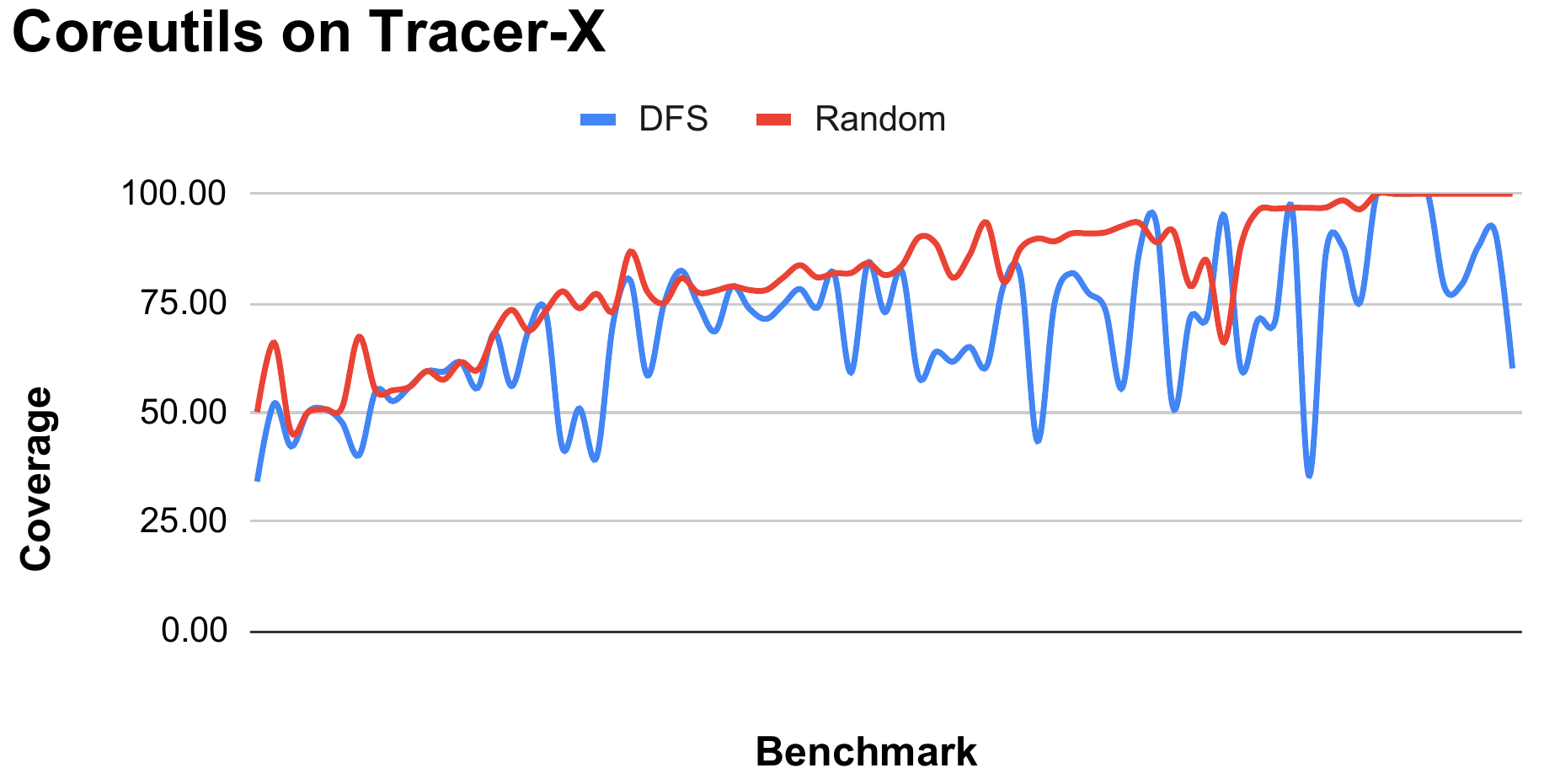}
		\caption{TracerX-D vs. TracerX-R}
		\label{fig:BBCovCor}
	\end{subfigure}
	\vspace{-3mm}
	\caption{Analysis of Coreutils Programs }
	\label{fig:Coreutils}
	\vspace{-5mm}
\end{figure}

\begin{figure}[!htb]
	\centering
	\begin{subfigure}{.5\textwidth}
		\centering
		\includegraphics[width=1\linewidth]{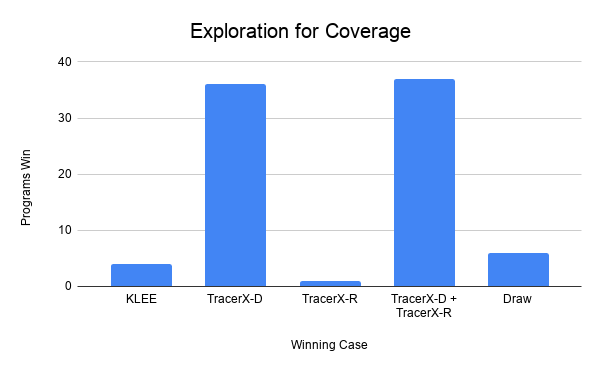}
		\caption{Exploration of Coverage for SV-COMP programs}
		\label{fig:NonCoreutils}
	\end{subfigure}%
	\begin{subfigure}{.5\textwidth}
		\centering
		\includegraphics[width=1\linewidth]{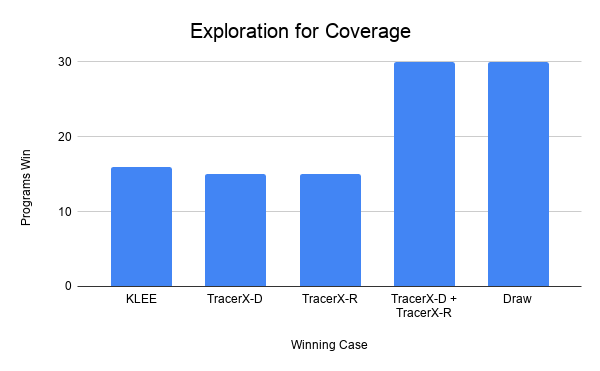}
		\caption{Exploration of Coverage for Coreutils programs}
		\label{fig:BBCovCor}
	\end{subfigure}
	\vspace{-3mm}
	\caption{Aggregated results of Coverage}
	\label{fig:explorationcoverage}
		\vspace{-5mm}
\end{figure}

%% file: summ_experiment.tex
%\vspace{-3mm}

% Tables: 16 60 27 = 103

We now discuss the results presented above.

\subsection{Main Experiment}

Fig. \ref{fig:AllTargetscount} considered all targets.
%which presents the aggregated number of targets where each
%tool was able to prove in the All targets experiment. 
Clearly \tracerx{} has superior results in terms of
proving both reachable and unreachable targets;
it times out less.
Note that KLEE was relatively poor in proving unreachable targets,
while CBMC was relatively poor for reachable targets.
In the end, \tracerx{} \emph{wins} in 1339 (26.57\%)  targets, while  \emph{loses} in only
112 (2.21\%) targets. 
Moving to hard targets, Fig. \ref{fig:HardTargetscount}, the gap widens.
\tracerx{} \emph{wins} in 796 (54.15\%)  targets, while  \emph{loses} in only
64 (4.35\%) targets.

\ignore{
It is  clear that \tracerx{} outperforms 
the baseline tools in all targets experiment since it has a smaller number of timeout cases. 
KLEE is almost negligible for proving unreachable targets since it has only proved 2 unreachable 
targets.  On the other hand, KLEE outperforms CBMC in finding reachable targets. 

See Fig. \ref{fig:HardTargetscount} which presents  the aggregated number of targets where each
tool was able to prove in the  Hard targets experiment. Again, it is  clear that \tracerx{} outperforms 
the baseline tools since it has a smaller number of timeout cases.  We also observed that the number 
of timeouts \tracerx{} are less for hard targets compared to the All targets experiment.
Similarly,  \tracerx{} has superior results as compared to KLEE and CBMC in terms of
proving both reachable and unreachable targets. 
This time KLEE has more timeouts.  However, KLEE outperforms CBMC in finding reachable targets. 
Regarding, the total number of unreachable targets \tracerx{} can prove more comparing
to CBMC. KLEE has not proved any unreachable targets. For the total number of reachable targets,
again \tracerx{} proved more compared to CBMC and KLEE. 
In total, \tracerx{} \emph{wins} in 796 (54.15\%)  targets, while  \emph{loses} to KLEE/CBMC in only
64 (4.35\%) targets. Clearly, the performance of \tracerx{}  has increased on hard targets. 
}

In summary for the main experiment, we now present a metric 
$AT\_Win\_Ratio$
for the final results
in Tables \ref{tab:alltargets} and \ref{tab:hardtargets}.
Let $AT\_Win$ denote the number of \tracerx{} wins,
and  $AT\_Lose$ for losses. 
We define

$$
AT\_Win\_Ratio\% = \frac{AT\_Win - AT\_Lose}{AT} = \frac{1339 - 112}{5058} = 24.25\% 
$$

\noindent
to capture our performance advantage in percentage terms.
Similarly, for hard targets:

$$
HT\_Win\_Ratio\% = \frac{HT\_Win - HT\_Lose}{HT} = \frac{796 - 64}{1470} = 49.79\% 
$$

\noindent
Clearly, \tracerx{} is more effective as the targets become harder.

In a second comparison, we consider the relative speed of the tools.
Before proceeding we mention the total time, in minutes, utilized for the three tools,
\tracerx, KLEE and CBMC was 7782, 19648 and 20105 respectively
for all targets.  For hard targets, the numbers are
2634, 14630 and 10630. 

In Fig. \ref{fig:AllTargetsSpeed}, we aggregate the relative speedup of \tracerx{} 
over KLEE and CBMC.   Recall that we are considering targets for which the tools terminate.
Over all targets, \tracerx{} is 38.55$\times$ faster than KLEE and 137.56$\times$ faster than 
CBMC.  Also, it can be observed that KLEE and CBMC have nearly the same total time. Regarding 
the speed computation, for \tracerx{} has in total, 33 winning programs, and 10 losing programs 
as compared to KLEE\footnote{There are some programs with 0.0 times faster, it means the time
difference is very small.}. Also, \tracerx{} has in total, 20 winning programs and 0 losing programs 
as compared to CBMC. 
When considering hard targets, Fig. \ref{fig:HardTargetsSpeed}, the numbers are as follows.
\tracerx{} is 490.26$\times$ faster than KLEE and 37.50$\times$ faster than CBMC.  
\tracerx{} has won in 7 winning programs and loses in no programs as compared to KLEE. Also, \tracerx{} 
wins in 20 programs and loses in 4 programs as compared to CBMC. 

%In Fig. \ref{fig:HardTargetsSpeed}, we present the aggregate on the relative speedup of \tracerx{} 
%over KLEE and CBMC for hard targets. It is observed that CBMC did better than KLEE, because CBMC 
%is able to prove  unreachable targets more effectively. In total, the execution time consumed by 
%\tracerx{} was 2634 min, whereas KLEE consumed 14630 min and CBMC 10630 min for the Hard target 
%experiment. In total, 

\subsection{Supplementary Experiment}

We first discuss the bug-finding results.
We can observe in Table \ref{tab:supplm1m2}\footnote{Three programs were finished
very fast assume (<0.009 sec) which have been highlighted with blue
color.} that KLEE found the targets easily for 44
programs out of 47 programs. KLEE timeouts on the remaining 3 programs since all the
targets were unreachable. But, when we ran the same
set of programs with hard targets, then KLEE timeouts on 32 programs,
proves 4 programs have unreachable targets, and proves a first target as
reachable for 11 programs. We can observe that KLEE struggles in proving
hard targets. On the other hand, the performance of CBMC for \emph{all} and
\emph{hard} targets experiments is almost the same except for few cases. 

While in the all targets experiment KLEE outperforms CBMC, 
CBMC has better performance  in the hard targets programs.
There are some programs which are draw where two of the tools
have the same performance.

Finally, we consider \tracerx{}. Here, we observe that \tracerx{}-D is having a 
good performance in nearly all the easy targets. However, in some cases,
it fails to reach the performance of KLEE. In these programs, we notice that 
\tracerx{}-R is competitive compared with KLEE. Moving to the hard targets,
 we observe that \tracerx{}-D has a better performance compared to  \tracerx{}-R.

See Fig. \ref{fig:Bugfinding} 
which presents  the aggregated number of programs where each
tool was able to prove in the  all targets and hard targets experiments. 
Here, we separately compare \tracerx{}-D and \tracerx{}-R with the baseline tools.
\tracerx{}-D clearly outperforms KLEE and CBMC in both all targets and 
hard target experiments. In the all targets experiment, we notice that
\tracerx{}-R wins on nearly as many programs as KLEE. By combining 
\tracerx{}-D and \tracerx{}-R for these experiments,   i.e. \tracerx{}-D
+ \tracerx{}-R, then \tracerx{} has more number of winning cases
and  outperforms KLEE and CBMC significantly.

In summary, we conclude  that \tracerx{}-D has
good performance on both all Targets and hard Targets categories
Also, we have noted that \tracerx{}-R is competitive with
KLEE and when considered it can improve the overall performance of \tracerx{}.

We now discuss the coverage experiments, where we compare with only KLEE,
and the set of targets is defined by the basic blocks.
Fig. \ref{fig:NonCoreutils} shows the aggregated results for SV-COMP
programs. It can be observed that KLEE terminated only on 5 programs
and timeout on 42 programs, whereas \tracerx{}-D terminated on 31
programs and \tracerx{}-R terminated on 13 programs. Among the 47
programs, KLEE wins on 4 programs. There is a group where none of the
systems terminated within timeout so higher \textit{BB} will be
required to compare. Here, \tracerx{}-D wins on 6 programs out of 15
 programs. There is 1 program for which \tracerx{}-R won, and also
for 2 programs \tracerx{}-R has better coverage compared to \tracerx{}-D
but the same as KLEE. If we consider \tracerx{}-D and \tracerx{}-R
together then our system wins on 38 programs.

We finally discuss  the performances of KLEE, \tracerx{}-D, 
and \tracerx{}-R on Coreutils benchmarks. From
Table \ref{tab:coreutils} in the Appendix section, we can observe
that \tracerx{}-D terminates and is faster in 12 programs compared to 
KLEE. 
Next, in Fig. \ref{fig:Coreutils} on the Coreutils  programs where
neither KLEE nor \tracerx{} terminates, we observe that  KLEE has better
coverage in 13 programs. Here,  \tracerx{}-D does not perform well because
of the huge execution SET. \tracerx{}-D has better coverage in only  3
programs. However, \tracerx{}-R has competitive results as compared to
KLEE. \tracerx{}-R has better coverage on 15 programs. Moreover,  in
Fig. \ref{fig:BBCovCor} we report the aggregated result on the 75 
Coreutils programs. Overall, KLEE wins on 16 programs, and our
combined result of \tracerx{}-D + \tracerx{}-R wins on 30 programs.

In summary, consider first the 47 SV-COMP programs.
In bug-finding, \tracerx{} wins on 25 programs considering all targets,
and on 32 programs considering hard targets.
In coverage, the win is 38.
Finally, for the Coreutils programs, the win is 30 out of 75, with a loss of 16.

The overall conclusion of these sets of experiments is that our algorithm
has \emph{significantly improved path coverage} of DSE by means of its interpolation algorithm.
Clear evidence is given by showing many targets where \tracerx{} complete
search while other systems cannot, or are significantly slower.
When faced with an incomplete search, the result is less clear.
This may be because the link between path coverage and 
code coverage/bug-finding is not clear.
Nevertheless, our experiments do show that our algorithm is competitive
or better for this purpose too.

%% file: related.tex
Abstraction learning in symbolic execution has its origin in
\cite{jaffar09intp}, and is also implemented in the TRACER
system~\cite{jaffar11unbounded,jaffar12tracer}. TRACER implements two
interpolation techniques: using unsatisfiability core and weakest
precondition (termed \emph{postconditioned\/} symbolic execution in
\cite{yi15postconditioned}). Systems that use unsatisfiability core
and weakest precondition respectively include Ultimate
Automizer~\cite{heizmann14ultimate}, and a KLEE modification reported
in \cite{yi15postconditioned}. The use of unsatisfiability core
results in an interpolant that is conjunctive for a given program
point and therefore requires less performance penalty in handling.
In contrast, 
%whereas this cannot be guaranteed with the use of weakest
%precondition, however, 
weakest precondition might be more expensive to compute, yet logically
is the weakest interpolant, hence its use may result in more
subsumptions.

%%% HIEP: I find describing no-good learning before lazy annotations
%%% breaks the flow.

\ignore{
Abstraction learning is related to various no-good learning techniques
in \emph{Constraint Satisfaction Problem\/} (\emph{CSP\/})
solvers~\cite{frost94deadend} and \emph{conflict-driven and clause
  learning\/} (\emph{CDCL\/}) techniques in SAT
solving~\cite{bayardo97sat,silva96grasp}.  In these, learning occurs
when an algorithm encounters a {\em dead-end}, where any further
decision in the decision tree will remain inconsistent.
}

Abstraction learning is also popularly known as \emph{lazy annotations\/} 
(LA) in \cite{mcmillan10lazy,mcmillan14lazy}. In \cite{mcmillan14lazy}
McMillan reported experiments on comparing abstraction learning with
various other approaches, including \emph{property-directed
  reachability\/} (\emph{PDR\/}) and \emph{bounded model checking\/}
(\emph{BMC\/}).  He observed that PDR, as implemented in Z3 produced
less effective learned annotations. On the other hand, BMC technology, 
e.g. \cite{clarke05sat,cordeiro12esbmc,holzer08fshell,llbmc12},
employs as backend a SAT or SMT solver, hence it employs learning,
however, its learning is \emph{unstructured}, where a learned clause
may come from the entire formula~\cite{mcmillan14lazy}. In contrast,
learning in LA is structured, where an interpolant
learnt is a set of facts describing a single program point.

\ignore{
\noindent
\fbox{\begin{minipage}{0.48\textwidth}
Sometimes the encoding done by BMC can slow down BMC's performance,
see the
example \url{https://github.com/tracer-x/klee-examples/blob/master/basic/malloc3.c},
which has different mallocs in both if branches within the loop.  When
the mallocs are hoisted outside the loop into a single malloc as
in \url{https://github.com/tracer-x/klee-examples/blob/master/basic/malloc2.c},
however, the performance issue disappear. Preliminary check reveals
that by changing the value of MAX, although the formula size increases linearly,
the time increases exponentially.

\noindent
MAX / Time (s) / Formula size (bytes of SMT-LIB2 file)\\
7   /      0.416   /       71362\\
9   /      2.144    /      104780\\
11  /      12.548   /      142998\\
13  /      68.784   /      186016
\end{minipage}}
}

\ignore{
Other than in abstraction learning, Craig interpolation has been used
in \emph{counter\-example-guided abstraction refinement\/}
(\emph{CEGAR\/})~\cite{henzinger04proof,beyer11cpachecker,dudka11predator,greben12hsfc,kroening11wolverine},
which is one of the major approaches in software verification.
Abstraction learning can be said to be another category different from
CEGAR: It computes abstraction of programs similar to CEGAR, but in
the \emph{reverse\/} order: it starts with the most refined
abstraction, and gradually reduces the amount of information.
}

Recently {\em Veritesting} \cite{avgerinos16veritesting} leveraged modern SMT
solvers to enhance symbolic execution for bug finding. 
Basically, a program is partitioned into {\em difficult} and
{\em easy} fragments: the former are explored
in DSE mode (i.e., KLEE mode), while the latter are explored using
SSE mode with some power of pruning (i.e., BMC mode).
Though this paper and {\em veritesting} share the same motivation,
the distinction is clear. First, our learning is \emph{structured} and 
has \emph{customizable} interpolation techniques. Second, we directly
address the problem of pruning in DSE mode via the use of symbolic addresses.
In contrast, there will be program fragments where Veritesting's performance
will downgrade to naive DSE, e.g. our motivating examples.
In summary, we believe that our proposed algorithm can also be 
used to enhance Veritesting.

Our approach is also slightly related to various \emph{state merging\/}
techniques in symbolic execution, in the sense that both state merging
and abstraction learning terminates a symbolic execution path
prematurely while ensuring precision. State merging encodes multiple
symbolic paths using \texttt{ite} expressions (disjunctions) fed into the
solver. The article \cite{hansen09merging} shows that state merging
may result in significant degradation of performance, which hints that
complete reliance on constraint solver for path exploration, as with
the bounded model checkers (e.g., CBMC, LLBMC), may not always be the
most efficient approach for symbolic
execution.
\ignore{
\cite{kuznetsov12efficient} proposes a symbolic execution
that is based on KLEE, addressing the two problems of state merging:
degradation of performance due to solver having to deal with
disjunctions, and inability to control search strategy (the strategy
is dictated by the solver's implementation), respectively using
heuristics.
}

Finally, there is very recent work on KLEE \cite{trabish2018chopped}
that exploits a dependency analysis to identify redundant code
fragments that may be ignored during symbolic execution.  More
specifically, they execute some user-chosen functions only on-demand,
using program slicing to reduce demand.
This work is somewhat orthogonal to our work because of the manual input
and because the slicing is a static process. 
In contrast, our algorithm is completely general and dynamic.

%% file: conclusion.tex
We presented a new interpolation algorithm and an
implementation \tracerx{} to extend KLEE with pruning.  The main
objective is to address the path explosion problem in pursuit
of \emph{code penetration}: to prove that a target program point is
either reachable or unreachable.  That is, our focus is verification.
We showed via a comprehensive experimental evaluation that, while
computing interpolants has a very expensive overhead, the pruning
it provides often far outweighs the expense, and
brings significant advantages. 
In the experiments, we compared
against KLEE, a dynamic symbolic system with no pruning, and CBMC, a
static symbolic execution system which does have pruning.
We showed that our system outperforms when experimented for penetration.
In fact, the performance gap widens when the verification target is
harder to prove.  We finally demonstrated that our system
is also competitive in testing.

%% file: appendix.tex
In this Appendix section, we present the 
detailed result of our experiments on GNU Coreutils benchmark. 
 Table \ref{tab:coreutils} has 5 major
columns.  Columns 1 and 2 show the \textit{Benchmark} name
and \textit{\#TB} i.e. total number of basic blocks in LLVM IR
respectively. Columns 3 to 5 show results of KLEE, \tracerx{}-D,
and \tracerx{}-R respectively. These columns further split into four
sub-columns each. These sub-columns
are \textit{\#Inst}, \textit{ \#T }, \textit{\#VB},
and \textit{\#err}. The \textit{\#Inst} (in Millions) shows the total
number of LLVM instructions covered during the exploration of
the SET. \textit{ \#T } (in seconds) shows the total amount of
execution time consumed. The timeout we set for this experiment was 1
hour. \textit{\#VB} shows the total number of uniquely visited basic
blocks. \textit{\#err} is the number of error paths traversed
during the execution. 

\footnotesize
\begin{longtable}[c]{|c|c|c|c|c|c|c|c|c|c|c|c|c|c|}
	\caption{The results of Analysis of GNU Coreutils Benchmarks}
	\label{tab:coreutils}\\
	\hline
	\multirow{2}{*}{Benchmark} & \multirow{2}{*}{\#TB} & \multicolumn{4}{c|}{KLEE}        & \multicolumn{4}{c|}{TracerX-D}   & \multicolumn{4}{c|}{TracerX-R}   \\ \cline{3-14} 
	&                       & \#Inst & \#T     & \#VB & \#err & \#Inst & \#T     & \#VB & \#err & \#Inst & \#T     & \#VB & \#err \\ \hline
	\endfirsthead
	\multicolumn{14}{c}%
	{{\bfseries Table \thetable\ continued from previous page}} \\
	\endhead
	basename                   & 51                    & 201.0  & 110.23   & 35   & 0     & 0.5    & 3.92     & 35   & 0     & 1.5    & 15.81    & 35   & 0     \\ \hline
	chroot                     & 20                    & 29.9   & 18.84    & 15   & 0     & 1.7    & 9.08     & 15   & 0     & 3.9    & 19.83    & 15   & 0     \\ \hline
	cksum                      & 52                    & 109.8  & 73.85    & 38   & 0     & 22.9   & 147.24   & 38   & 0     & 34.4   & 194.18   & 38   & 0     \\ \hline
	dirname                    & 30                    & 55.0   & 35.5     & 15   & 0     & 0.7    & 6.21     & 15   & 0     & 3.0    & 17.91    & 15   & 0     \\ \hline
	printenv                   & 33                    & 748.5  & 463.82   & 26   & 0     & 0.9    & 4.03     & 26   & 0     & 1.5    & 10.39    & 26   & 0     \\ \hline
	pwd                        & 70                    & 2.3    & 1.49     & 43   & 0     & 0.6    & 1.73     & 43   & 0     & 0.6    & 2.37     & 43   & 0     \\ \hline
	runcon                     & 75                    & 51.5   & 82.99    & 38   & 0     & 6.4    & 170.76   & 38   & 0     & 39.9   & 1902.58  & 38   & 0     \\ \hline
	sync                       & 11                    & 1.8    & 1.11     & 9    & 0     & 0.2    & 0.88     & 9    & 0     & 0.3    & 1.45     & 9    & 0     \\ \hline
	uptime                     & 62                    & 53.5   & 40.84    & 34   & 0     & 3.0    & 12.14    & 34   & 0     & 6.9    & 30.18    & 34   & 0     \\ \hline
	users                      & 32                    & 46.9   & 37.06    & 19   & 0     & 0.9    & 5.76     & 19   & 0     & 4.9    & 23.54    & 19   & 0     \\ \hline
	echo                       & 66                    & 4344.5 & $\infty$ & 45   & 0     & 0.5    & 3.75     & 45   & 0     & 1.7    & 27.16    & 45   & 0     \\ \hline
	sum                        & 52                    & 4618.3 & $\infty$ & 52   & 0     & 19.4   & 86.07    & 52   & 0     & 73.2   & 360.86   & 52   & 0     \\ \hline
	tee                        & 44                    & 4604.6 & $\infty$ & 44   & 0     & 9.7    & 35.97    & 44   & 0     & 65.5   & 309.17   & 44   & 0     \\ \hline
	pinky                      & 120                   & 1995.1 & $\infty$ & 67   & c     & 40.8   & 234.3    & 67   & 0     & 68.0   & $\infty$ & 67   & c     \\ \hline
	env                        & 28                    & 4610.7 & $\infty$ & 28   & 0     & 22.5   & 65.79    & 28   & 0     & 133.4  & $\infty$ & 28   & 0     \\ \hline
	hostname                   & 15                    & 1233.2 & $\infty$ & 14   & 0     & 154.8  & $\infty$ & 13   & 0     & 52.5   & $\infty$ & 14   & 0     \\ \hline
	unlink                     & 13                    & 523.3  & $\infty$ & 13   & c     & 160.2  & $\infty$ & 13   & 0     & 54.2   & $\infty$ & 13   & 0     \\ \hline
	base64                     & 249                   & 11.7   & $\infty$ & 208  & 0     & 10.6   & $\infty$ & 205  & 0     & 39.4   & $\infty$ & 208  & 0     \\ \hline
	cp                         & 242                   & 17.3   & $\infty$ & 107  & 2     & 106.5  & $\infty$ & 102  & 4     & 46.7   & $\infty$ & 110  & c     \\ \hline
	cat                        & 119                   & 226.9  & $\infty$ & 106  & 1     & 37.1   & $\infty$ & 89   & 9337  & 73.4   & $\infty$ & 106  & 14964 \\ \hline
	chcon                      & 96                    & 398.0  & $\infty$ & 71   & 0     & 7.1    & $\infty$ & 38   & 0     & 33.8   & $\infty$ & 74   & c     \\ \hline
	chgrp                      & 59                    & 73.3   & $\infty$ & 49   & 0     & 71.2   & $\infty$ & 43   & 32    & 57.1   & $\infty$ & 48   & c     \\ \hline
	chmod                      & 123                   & 66.5   & $\infty$ & 91   & 2     & 71.6   & $\infty$ & 87   & 740   & 56.3   & $\infty$ & 90   & c     \\ \hline
	comm                       & 66                    & 861.0  & $\infty$ & 65   & 0     & 170.1  & $\infty$ & 58   & 0     & 94.4   & $\infty$ & 65   & c     \\ \hline
	chown                      & 44                    & 99.3   & $\infty$ & 36   & 0     & 163.1  & $\infty$ & 26   & 0     & 63.9   & $\infty$ & 36   & 0     \\ \hline
	csplit                     & 375                   & 65.2   & $\infty$ & 278  & 0     & 71.3   & $\infty$ & 219  & 0     & 57.9   & $\infty$ & 291  & c     \\ \hline
	cut                        & 201                   & 494.7  & $\infty$ & 188  & 2     & 144.8  & $\infty$ & 144  & 1     & 19.4   & $\infty$ & 170  & 30    \\ \hline
	date                       & 100                   & 23.2   & $\infty$ & 85   & 0     & 9.0    & $\infty$ & 58   & 0     & 33.6   & $\infty$ & 90   & c     \\ \hline
	df                         & 209                   & 293.3  & $\infty$ & 182  & 2     & 104.7  & $\infty$ & 165  & 1713  & 64.3   & $\infty$ & 167  & c     \\ \hline
	du                         & 216                   & 11.5   & $\infty$ & 132  & 0     & 27.9   & $\infty$ & 128  & 0     & 4.5    & $\infty$ & 124  & c     \\ \hline
	expand                     & 95                    & 1369.3 & $\infty$ & 92   & 0     & 59.4   & $\infty$ & 82   & 0     & 100.1  & $\infty$ & 92   & c     \\ \hline
	factor                     & 57                    & 6.8    & $\infty$ & 49   & 0     & 1.3    & $\infty$ & 37   & 0     & 2.6    & $\infty$ & 49   & 0     \\ \hline
	fmt                        & 185                   & 100.9  & $\infty$ & 158  & 0     & 10.6   & $\infty$ & 118  & 0     & 93.7   & $\infty$ & 164  & c     \\ \hline
	fold                       & 91                    & 79.7   & $\infty$ & 79   & 0     & 115.3  & $\infty$ & 55   & 0     & 48.9   & $\infty$ & 85   & 0     \\ \hline
	head                       & 214                   & 165.0  & $\infty$ & 173  & 2     & 74.4   & $\infty$ & 158  & 5     & 34.9   & $\infty$ & 173  & 79    \\ \hline
	hostid                     & 16                    & 433.8  & $\infty$ & 14   & 1     & 154.3  & $\infty$ & 13   & 1     & 55.9   & $\infty$ & 14   & 3     \\ \hline
	id                         & 78                    & 3093.3 & $\infty$ & 67   & c     & 83.1   & $\infty$ & 48   & 0     & 72.6   & $\infty$ & 63   & c     \\ \hline
	join                       & 385                   & 52.2   & $\infty$ & 311  & 0     & 44.5   & $\infty$ & 301  & 0     & 52.6   & $\infty$ & 322  & 0     \\ \hline
	kill                       & 141                   & 21.3   & $\infty$ & 132  & 0     & 63.2   & $\infty$ & 134  & 12    & 1.4    & $\infty$ & 93   & 0     \\ \hline
	link                       & 14                    & 463.6  & $\infty$ & 14   & c     & 159.7  & $\infty$ & 11   & 0     & 87.1   & $\infty$ & 14   & 0     \\ \hline
	ln                         & 176                   & 183.5  & $\infty$ & 134  & 0     & 111.4  & $\infty$ & 145  & 0     & 56.8   & $\infty$ & 142  & 0     \\ \hline
	logname                    & 11                    & 364.1  & $\infty$ & 10   & 0     & 153.5  & $\infty$ & 9    & 0     & 57.1   & $\infty$ & 10   & 0     \\ \hline
	mkdir                      & 28                    & 438.1  & $\infty$ & 28   & 1     & 169.0  & $\infty$ & 21   & 0     & 70.1   & $\infty$ & 27   & c     \\ \hline
	mkfifo                     & 24                    & 535.6  & $\infty$ & 24   & 1     & 173.4  & $\infty$ & 19   & 1594  & 85.1   & $\infty$ & 24   & c     \\ \hline
	mktemp                     & 52                    & 357.9  & $\infty$ & 50   & 0     & 176.3  & $\infty$ & 37   & 0     & 77.7   & $\infty$ & 50   & c     \\ \hline
	mv                         & 108                   & 46.4   & $\infty$ & 84   & 0     & 12.7   & $\infty$ & 74   & 0     & 52.5   & $\infty$ & 84   & 0     \\ \hline
	nice                       & 27                    & 32.2   & $\infty$ & 25   & 0     & 2.9    & $\infty$ & 15   & 0     & 15.4   & $\infty$ & 25   & 0     \\ \hline
	nl                         & 102                   & 484.6  & $\infty$ & 96   & 1     & 164.1  & $\infty$ & 61   & 0     & 93.8   & $\infty$ & 90   & c     \\ \hline
	nohup                      & 40                    & 1546.1 & $\infty$ & 22   & 0     & 154.4  & $\infty$ & 21   & 0     & 75.7   & $\infty$ & 22   & c     \\ \hline
	paste                      & 158                   & 149.9  & $\infty$ & 144  & 0     & 64.7   & $\infty$ & 116  & 0     & 53.5   & $\infty$ & 144  & c     \\ \hline
	ptx                        & 704                   & 3.9    & $\infty$ & 175  & 0     & 19.7   & $\infty$ & 240  & 0     & 20.2   & $\infty$ & 352  & 0     \\ \hline
	readlink                   & 31                    & 278.3  & $\infty$ & 30   & c     & 7.2    & $\infty$ & 30   & 0     & 40.8   & $\infty$ & 30   & 0     \\ \hline
	rm                         & 45                    & 235.0  & $\infty$ & 42   & 0     & 145.6  & $\infty$ & 42   & 0     & 60.6   & $\infty$ & 40   & c     \\ \hline
	rmdir                      & 61                    & 72.8   & $\infty$ & 45   & 0     & 110.7  & $\infty$ & 31   & 0     & 59.3   & $\infty$ & 45   & c     \\ \hline
	seq                        & 106                   & 97.8   & $\infty$ & 99   & 0     & 0.6    & $\infty$ & 54   & c     & 6.7    & $\infty$ & 97   & 0     \\ \hline
	setuidgid                  & 31                    & 350.4  & $\infty$ & 30   & 0     & 103.1  & $\infty$ & 11   & 0     & 75.2   & $\infty$ & 30   & 0     \\ \hline
	shuf                       & 136                   & 9.8    & $\infty$ & 109  & 0     & 10.7   & $\infty$ & 97   & 0     & 4.3    & $\infty$ & 106  & 0     \\ \hline
	shred                      & 296                   & 25.8   & $\infty$ & 151  & 0     & 27.6   & $\infty$ & 141  & 0     & 16.0   & $\infty$ & 151  & c     \\ \hline
	sleep                      & 25                    & 213.4  & $\infty$ & 25   & 0     & 29.7   & $\infty$ & 22   & 0     & 15.1   & $\infty$ & 25   & c     \\ \hline
	sort                       & 689                   & 2.1    & $\infty$ & 271  & 0     & 2.7    & $\infty$ & 358  & 0     & 12.1   & $\infty$ & 454  & 0     \\ \hline
	split                      & 166                   & 68.2   & $\infty$ & 123  & 1     & 26.7   & $\infty$ & 133  & 17458 & 57.0   & $\infty$ & 144  & 162   \\ \hline
	stat                       & 252                   & 135.2  & $\infty$ & 172  & 2     & 24.8   & $\infty$ & 141  & 126   & 57.7   & $\infty$ & 185  & 61    \\ \hline
	stty                       & 434                   & 32.5   & $\infty$ & 273  & 0     & 106.6  & $\infty$ & 241  & 0     & 24.3   & $\infty$ & 259  & 0     \\ \hline
	tac                        & 118                   & 270.5  & $\infty$ & 93   & 1     & 4.0    & $\infty$ & 87   & 0     & 36.3   & $\infty$ & 92   & 0     \\ \hline
	tail                       & 510                   & 69.0   & $\infty$ & 375  & 2     & 2.9    & $\infty$ & 213  & 0     & 51.1   & $\infty$ & 396  & 1     \\ \hline
	touch                      & 130                   & 97.7   & $\infty$ & 105  & 2     & 5.1    & $\infty$ & 97   & 6     & 22.6   & $\infty$ & 105  & 79    \\ \hline
	tr                         & 586                   & 49.0   & $\infty$ & 321  & 0     & 118.5  & $\infty$ & 235  & 0     & 41.1   & $\infty$ & 394  & 0     \\ \hline
	tsort                      & 123                   & 19.1   & $\infty$ & 115  & 0     & 0.9    & $\infty$ & 88   & 0     & 6.3    & $\infty$ & 97   & 0     \\ \hline
	unexpand                   & 116                   & 138.8  & $\infty$ & 112  & 1     & 109.7  & $\infty$ & 82   & 0     & 71.6   & $\infty$ & 112  & c     \\ \hline
	uniq                       & 166                   & 63.0   & $\infty$ & 147  & 0     & 6.9    & $\infty$ & 72   & 0     & 41.4   & $\infty$ & 149  & 0     \\ \hline
	tty                        & 22                    & 80.4   & $\infty$ & 20   & 0     & 101.1  & $\infty$ & 17   & 0     & 60.5   & $\infty$ & 20   & c     \\ \hline
	who                        & 150                   & 234.1  & $\infty$ & 116  & c     & 92.5   & $\infty$ & 112  & 0     & 65.4   & $\infty$ & 116  & 0     \\ \hline
	wc                         & 151                   & 739.6  & $\infty$ & 125  & c     & 127.8  & $\infty$ & 127  & 0     & 50.5   & $\infty$ & 127  & c     \\ \hline
	whoami                     & 11                    & 779.1  & $\infty$ & 11   & 0     & 160.2  & $\infty$ & 10   & 0     & 51.5   & $\infty$ & 11   & 0     \\ \hline
	yes                        & 15                    & 490.9  & $\infty$ & 15   & 0     & 1.5    & $\infty$ & 9    & 0     & 41.6   & $\infty$ & 15   & 0     \\ \hline
\end{longtable}